\documentclass[11pt]{article}
\usepackage{amsfonts}
\usepackage{amsmath,amssymb,amsthm}
\usepackage{hyperref}
\usepackage{graphicx}
\usepackage{mathrsfs}  
\usepackage{enumitem}
\usepackage{float}
\usepackage[]{algorithm2e}
\usepackage{color}
\usepackage[margin=1in]{geometry}
\usepackage{bbm}
\usepackage{comment}
\usepackage{soul}
\graphicspath{{figures/}}
\numberwithin{equation}{section}

\renewenvironment{thebibliography}[1]{\begin{oldthebibliography}{#1}\setlength{\baselineskip}{.8em}
\linespread{.8}
\small
\setlength{\parskip}{0ex}\setlength{\itemsep}{.1em}}{\end{oldthebibliography}}
\theoremstyle{plain}
\newtheorem{theorem}{Theorem}[section]

\newtheorem{lemma}[theorem]{Lemma}
\newtheorem{corollary}[theorem]{Corollary}

\DeclareMathAlphabet\scr{U}{scr}{m}{n}
\SetMathAlphabet\scr{bold}{U}{scr}{b}{n}
  \DeclareFontFamily{U}{scr}{\skewchar\font'177}
  \DeclareFontShape{U}{scr}{m}{n}{<-6>rsfs5<6-8>rsfs7<8->rsfs10}{}
  \DeclareFontShape{U}{scr}{b}{n}{<-6>rsfs5<6-8>rsfs7<8->rsfs10}{}
\theoremstyle{definition}
\newtheorem{remark}[theorem]{Remark}
\newtheorem{assumption}[theorem]{Assumption}

\def\E{\mathbb{E}}
\def\R{\mathbb{R}}
\def\Hp{\mathbb{H}}
\def\sign{\text{sign}}

\title{ Deep Learning Algorithms for Hedging with Frictions\thanks{The authors thank the two anonymous reviewers for their careful readings and suggestions, and thank Lukas Gonon, Jiequn Han, Ruimeng Hu, Steven Kou, Johannes Muhle-Karbe, Junru Shao, Mete Soner, and Xunyu Zhou for useful comments and fruitful discussions. Part of the work is supported by the MA Mentored Research Program in Statistics at Columbia University. }}
\date{}
\author{
Xiaofei Shi\thanks{ University of Toronto, Department of Statistical Sciences, email \url{xf.shi@utoronto.ca}.}
\and
Daran Xu \thanks{Columbia University, Department of Statistics, email \url{dx2207@columbia.edu}.}
\and
Zhanhao Zhang \thanks{Columbia University, Department of Statistics, email \url{zz2760@columbia.edu}.}
}
\begin{document}
\maketitle

\begin{abstract}
This work studies the deep learning-based numerical algorithms for optimal hedging problems in markets with general convex transaction costs on the trading rates, focusing on their scalability of trading time horizon. Based on the comparison results of the FBSDE solver by Han, Jentzen, and E (2018) and the Deep Hedging algorithm by Buehler, Gonon, Teichmann, and Wood (2019), we propose a Stable Transfer Hedging (ST-Hedging) algorithm, to aggregate the convenience of the leading-order approximation formulas and the accuracy of the deep learning-based algorithms. Our ST-Hedging algorithm achieves the same state-of-the-art performance  in short  and moderately long time horizon as FBSDE solver and Deep Hedging, and generalize well to long time horizon when previous algorithms become suboptimal. With the transfer learning technique, ST-Hedging drastically reduce the training time, and  shows great scalability to high-dimensional settings. This opens up new possibilities in model-based deep learning algorithms in economics, finance, and operational research, which takes advantages of the domain expert knowledge and the accuracy of the learning-based methods. 
\end{abstract}

\section{Introduction}

As observed in many empirical papers, markets are imperfect, meaning that arbitrary quantities cannot be traded immediately at the quoted market price because of taxes, regulations, and the limited liquidity of the assets. Typical examples include linear transaction taxes as well as fixed transaction costs. As reported and  studied in~\cite{almgren.03,almgren.chriss.01,almgren.al.05,lillo.al.03}, empirical estimates of actual transaction costs typically correspond to a 3/2-th power of the order flow.  Accordingly, the large trading volume quickly impacts the market liquidity, which, in turn, drastically changes the agents' behaviors. Hence, optimally scheduling the order flow in anticipation of market liquidity shortage is crucial. 

There is a large body of literature on optimal execution strategies as well as dynamic portfolio optimization models with market illiquidity and price impacts, see, e.g.~\cite{amihud.al.06,bouchaud.al.12,dumas.luciano.91,janecek.shreve.10, kohlmann.tang.02, liu2004optimal,shreve1994optimal}, and with recent results in~\cite{bank.al.17,garleanu.pedersen.13, guasoni.weber.18,soner.touzi.13}. 
By assuming the transaction costs as a quadratic function of the size of the order flow, the optimal trading policy can be given in a closed-form, as shown in~\cite{garleanu.pedersen.13,kohlmann.tang.02}. 
However, no closed-form solution is available under general nonlinear transaction costs with general market dynamics.   
To obtain tractable results, researchers then focus on the  \emph{small} costs limit, as in~\cite{almgren.li.16,bayraktar.al.18, guasoni.weber.18,kallsen.muhlekarbe.17,moreau.al.17,shreve1994optimal,soner.touzi.13}. These elegant asymptotic formulas were proved 
rigorously by~\cite{Ahrens2015OnUS,herdegen2018stability,kallsen2013portfolio,shreve1994optimal}. 
Two important follow-up questions therefore arise. 
The first one is about the quality of the leading-order approximation, which needs to be well-studied and quantified, especially in empirical examples. The second is about the smallness assumption, whether that is an absolute or a relative quantity. 

From the numerical perspectives, recent advances in highly accurate machine learning models have introduced powerful new tools for studying high-dimensional optimization problems, such as hedging with frictions.
The \emph{FBSDE solver}, developed by Han, Jentzen, and E in~\cite{han.al.17}, can solve a dynamic portfolio optimization problem by  finding the solution to a BSDE system. 
For the first time, this algorithm overcame the curse of dimensionality in numerical solutions to high-dimensional SDE and associated PDE, as pointed out by~\cite{beck2020overview,Grohs2018}. The convergence analysis is established by~\cite{han2020convergence} under the same short-term existence assumptions in~\cite{delarue2002existence}. 
At a higher level, the FBSDE solvers find the solution of the FBSDEs through a \emph{supervised learning} framework, since the accuracy of the terminal value of the backward components is served as the goal functional of the algorithm.  
 However, this algorithm does not scale well with the trading time horizon and the time discretization. For calibrated trading parameters as in~\cite{gonon2021asset}, the time horizon for the algorithm to work is required to be unreasonably small. 
In the meantime, with the development of modern model-free techniques, \emph{reinforcement learning} algorithms are also widely used in single-agent optimization problems. Indeed, as shown in the groundbreaking  \emph{Deep Hedging} algorithm~\cite{buehler2019deep,buehler2019deepa}, and similar works in the same flavor~\cite{becker2019deep,casgrain2019deep,han2016deep,hu2019deep, hure2018deep, min2021signatured, Moallemi2021RL,sun2020optimizing,reppen2020bias,ruf2021hedging}, we treat the utility functions as targets and directly parametrize and learn the optimal trading policy. 
Moreover, reinforcement learning frameworks are introduced and analyzed rigorously in ~\cite{wang2020reinforcement, wang2020continuous}. 
Therefore, just like in~\cite{Moallemi2021RL}, a natural question is how these methods compare in practice, and our goal is to understand and compare the two types of methods, especially the advantages and drawbacks of supervised learning and reinforcement learning algorithms.

This paper studies the optimal hedging problem in frictional markets in a general setting and examine the machine learning-based numerical algorithms. 
First, we explore popular machine learning architectures, including the FBSDE solver and the Deep Hedging algorithms. 
We implement both methods, document the tuning procedures, and discuss their advantages and disadvantages.  
Then we compare our numerical results with the leading-order approximation formulas. 
These comparisons help justify the accuracy of the approximation trading formulas and provide us with ideas on the usage of each method in practice. 
In summary,  the FBSDE solver only performs well under short trading horizons, while it fails to work in intermediate trading horizons. 
The Deep Hedging algorithm has stable and reliable performance for short and moderately long trading horizons, but it does not scale well when the trading horizon becomes long. 
In contrast, the leading-order approximations work well under long trading horizons, while it deviates from the ground truth by a significant amount shortly before maturity. 
To fully utilize the convenience of the leading-order formulas and the accuracy of the learning-based algorithms, we employ the \emph{transfer learning} ideas and 
propose a new algorithm which we referred to as \emph{Stable Transfer Hedging} (ST-Hedging) algorithm. 
ST-Hedging algorithm achieves the same state-of-the-art performance in short time horizons by the FBSDE solver and in intermediate time horizons by the Deep Hedging algorithms, and 
it allows us to directly work with long trading horizons. 
ST-Hedging also shows great scalability for multiple stocks with correlations. 
In addition, the convenience of implementation and convergence speed of ST-Hedging algorithm outperforms the FBSDE solver and the Deep Hedging algorithms, which shows the great potential of model-based learning algorithms. 

The rest of the paper is organized as follows. First, in Section~\ref{sec:model}, we introduce the market model with frictions, the preference for individual agents, and the admissible strategies. Then, 
the FBSDE solver, Deep Hedging and ST-Hedging algorithms are introduced in Section~\ref{ML}, with details on the implementations and comparisons. 
Finally, we compare the  performance of different learning-based algorithms and the
leading-order approximation formula in Section~\ref{results}. 
For better readability, the derivation of leading-order approximations and all proofs are collected in Appendix~\ref{sec:asymptotic}.

\paragraph{Notation.}
We fix a filtered probability space $(\Omega,\mathscr{F},(\mathscr{F}_t)_{t \in [0,T]},\mathbb{P})$ with finite time horizon $T>0$, where the filtration is generated by a $1$-dimensional standard Brownian motion $W=(W_t)_{t \in [0,T]}$. 

\section{Model Setup}\label{sec:model}
Here, we 
assume the randomness in the market is generated by the
$1$-dimensional Brownian motion $(W_t)_{t \in [0,T]}$. Suppose we receive a cumulative random endowment as 
\begin{align*}
d\zeta_t =\xi_t dW_t, \quad \mbox{for a general process $\xi \in \Hp^2$.}
\end{align*}
There are two assets in the market for us to hedge the risk, 
the first one is safe, with price normalized to one, and the second is risky, with general dynamic
\begin{align}\label{eq:dyn2}
dS_t=\mu_t dt+\sigma_t dW_t,
\end{align}
where the expected return process satisfies the no-arbitrage condition $\mu= \sigma\kappa$ for a market price of risk $\kappa\in\Hp^2$. One readily verifies that the Bachelier models and Geometric Brownian motions satisfy this requirement. 

\paragraph{Transaction costs. }


A popular class of models originating from the optimal execution literature focuses on absolutely continuous trading strategies, cf.~\cite{almgren.03,almgren.chriss.01},
\begin{align*}
\varphi_t=\varphi_{0-}+\int_0^t \dot{\varphi}_u du, \quad t \geq 0,
\end{align*}
and penalizes the trading rate $\dot{\varphi}_t=d\varphi_t/dt$ with an \emph{instantaneous transaction cost} $ \lambda_t G(\dot{\varphi}_t)$. With proper definition of matrix-valued operation,  this formulation of transaction costs can be generalized to multiple risky assets with cross-asset costs. 

$G$ is used to model the ``shape'' of transaction costs. Portfolio choice problems for the most tractable quadratic specification $ G(x)= x^2/2$ are analyzed in models by \cite{almgren.li.16,bouchard.al.18,garleanu.pedersen.16,guasoni.weber.17,moreau.al.17,sannikov.skrzypacz.16}. In the small transaction costs limit as in~\cite{bayraktar.al.18,caye2017trading,guasoni.weber.18}, single-agent models are solved explicitly for the more general power costs $ G(x)= |x|^q/q$, $q \in (1,2]$ proposed by \cite{almgren.03}. Below, we introduce the general smooth convex cost functions $G$ as studied in~\cite{gonon2021asset,guasoni.rasonyi.15}, which includes the special examples mentioned above:

\begin{assumption}\label{cond:cost}
\begin{enumerate}[label=(\roman*)]
\item The transaction cost $G: \mathbb{R} \to \mathbb{R}_+$ is convex, symmetric, and strictly increasing on $[0,\infty)$, differentiable on $[0,\infty)$, and satisfies $G(0)=0$;
\item The derivative $G'$ is  also strictly increasing and differentiable on $(0,\infty)$ with $G'(0)=0$;
\item There exist constants $C>0$, $K\geq2$ and $x_0>0$ such that 
\begin{align*}
|(G')^{-1}(x)| \leq C(1+|x|^{K-1})\ \mbox{for all $x\in\R$}, \quad 
G''(x) \leq C \ \mbox{for all $|x|>x_0$}.
\end{align*}
\end{enumerate}
\end{assumption}
The power costs and their linear combinations are included, and the proportional costs can also be studied as the singular limit of power
costs with the power approaching 1.

The process $\lambda$ models the level of the transaction costs parameter. Like in the partial-equilibrium model of \cite{moreau.al.17}, we allow the transaction cost to fluctuate randomly over time:
\begin{align}\label{def:liquidity risk}
\lambda_t = \lambda \Lambda_t, \qquad t \in [0,T].
\end{align}
Here, the constant $\lambda>0$ modulates the magnitude of the cost (this scaling parameter will be sent to zero in the small-cost asymptotics, where explicit formula is available as in~\cite{guasoni.weber.18, caye.al.18}. We restate this results in Appendix~\ref{formal approximations}). The \emph{strictly positive} stochastic processes $(\Lambda_t)_{t \in [0,T]}$ describes the fluctuations of liquidity over time, and allows to model ``liqudity risk'' as in \cite{acharya.pedersen.05}, for example.

\paragraph{Preferences and Admissible Strategies. }

In order to capture the risk-aversion, we consider the \emph{linear-quadratic} model, i.e.,
we maximize our expected returns while penalizing for the corresponding quadratic variations: 
\begin{align}
J_T(\dot\varphi) =\frac{1}{T}\E\left[\int_0^T \Big(\varphi_t \mu_t -\frac{\gamma}{2}\left(\sigma_t\varphi_t+\xi_t\right)^2 - {\lambda_t}G\left(\dot{\varphi}_t\right)\Big)dt\right]. \label{general:regular goal}
\end{align}
As studied in~\cite{choi.larsen.15,garleanu.pedersen.13,sannikov.skrzypacz.16}, we trade off expected returns against the tracking error relative to the exogenous target position $-\xi/\sigma$.  We can also consider more sophisticated preferences such as exponential or power utility. 
Notice that the deep reinforcement learning algorithms such as Deep Hedging, can be easily generalized to sophisticated preferences or utility functions, but the usage of FBSDE solvers is limited since there might not exist a FBSDE system to describe the optimal solution of the hedging problem. See Section~\ref{sub: comparison} for details.

To make sure all terms are well defined, we focus on \emph{admissible} strategies that satisfy the following integrability conditions:
\begin{align}
\sup_{T>0}\frac{1}{T}\E\left[ \int_0^T\gamma \sigma_t^2 \varphi_t^2 + {\lambda_t}G\left(\dot{\varphi}_t\right)dt\right] <\infty.
\end{align}
Moreover, we impose the transversality condition to exclude Ponzi scheme, by ruling out arbitrarily large risky positions:
\begin{align}\label{assump:vanishing}
\lim_{\sqrt{\lambda}/T\to 0} \frac{\lambda}{T^2}\E\left[ \Lambda_T\sigma_T^2\varphi_T^2\right]  = 0. 
\end{align}

By strict convexity of the goal functional~\eqref{general:regular goal}, optimality of a trading rate $\dot{\varphi}$ is equivalent to the first-order condition that the Gateaux derivative $\lim_{\rho \to 0} (J_T(\dot\varphi+\rho\dot\eta)-J_T(\dot\varphi))/{\rho}$ vanishes for \emph{any} admissible perturbation $\eta$, cf.~\cite{ekeland.temam.99}:
\begin{align*}
0=\mathbb{E}_t\left[\int_0^T \left(\mu_t \int_0^t \dot{\eta}_u du-\gamma(\sigma_u\varphi_u+\xi_u)\sigma_u \int_0^t \dot{\eta}_u du - \lambda_t G'(\dot{\varphi}_t)\dot{\eta}_t\right)dt\right].
\end{align*}
As in \cite{bank.al.17}, this can be rewritten using Fubini's theorem as 
\begin{align*}
0=\mathbb{E}_t\left[\int_0^T \left(\int_t^T \Big(\mu_u -\gamma(\sigma_u\varphi_u+\xi_u)\sigma_u\Big) du -\lambda_tG'(\dot{\varphi}_t)\right)\dot{\eta}_t dt\right].
\end{align*}
Since this has to hold for \emph{any} perturbation $\dot{\eta}_t$, the tower property of conditional expectation yields
\begin{align}\label{eq:condition}
\lambda_tG'\left(\dot{\varphi}_t\right) &= \E_t\Big[\int_t^T \big(\mu_u -\gamma(\sigma_u\varphi_u+\xi_u)\sigma_u \big)du\Big]
= M_t
 + \int_0^t \big(\gamma(\sigma_u\varphi_u+\xi_u)\sigma_u -\mu_u\big)du,
\end{align}
for a martingale $dM_t={Z}_t dW_t$ that needs to be determined as part of the solution. Solving for the dynamics of the agents' optimal trading rates would introduce the dynamics of the transaction costs. Accordingly, it is preferable to instead work with the marginal transaction costs transaction costs as the backward process that describes the optimal controls: 
\begin{align}\label{eq:defY3}
Y_t := \lambda_tG'\left(\dot{\varphi}_t\right),
\end{align}
and from~\eqref{eq:condition} we can infer that $Y_T = 0$.
With this notation, the corresponding trading rates are
\begin{align*}
\dot{\varphi}_t = (G')^{-1}\left(\frac{Y_t}{\lambda_t}\right).
\end{align*}
Thus, the optimal position $\varphi$ and the corresponding marginal transaction costs $Y$ in turn solve the nonlinear FBSDE
\begin{align}
d\varphi_t &=(G')^{-1}\left(\frac{Y_t}{\lambda_t}\right) dt, && \varphi_0=\varphi_{0-},  \label{eq:eqphi}\\
dY_t &= \big(\gamma(\sigma_t\varphi_t+\xi_t)\sigma_t -\mu_t\big)dt + Z_t dW_t, &&Y_T =0.\label{eq:bwopt}
\end{align}
For constant quadratic costs $\lambda x^2/2$ and constant volatility $\sigma$, this FBSDE becomes linear and can in turn be solved by reducing it to a standard Riccati equation system~\cite{bank.al.17,bouchard.al.18,muhlekarbe2021equilibrium}. For volatilities and quadratic costs that fluctuate randomly, these ODEs are replaced by a backward \emph{stochastic} Riccati equation, compare~\cite{annkirchner.kruse.15,kohlmann.tang.02}. With nonlinear costs, no such simplifications are possible. In fact, the wellposedness of the system is generally unclear even for short time horizons since no Lipschitz condition for $(G')^{-1}$ is satisfied.

\paragraph{Reparametrization. }In the frictionless analogue of~\eqref{general:regular goal}, pointwise maximization yields the agent's individually optimal strategies, i.e.,
\begin{align}\label{eq:strat}
\bar{\varphi}_t = \frac{\mu_t}{\gamma\sigma_t^2} -\frac{\xi_t}{\sigma_t}, \qquad t \in [0,T].
\end{align}
In particular, we can write the dynamic of the frictionless strategy as
\begin{align}\label{eq:frictionless dynamic}
d \bar{\varphi}_t = \bar{b}_t dt + \bar{a}_t dW_t. 
\end{align}

In real-world applications, people are more interested in the changes induced by transaction costs relative to the frictionless version of the model. To facilitate both the numerical and the analytical analysis, we define 
\begin{align}
\Delta\varphi_t:=\varphi_t-\bar{\varphi}_t \label{eq: d delta phi}
\end{align}
for the difference between the frictional and frictionless positions, which we expect to vanish as the transaction costs tend to zero. In view of the dynamic~\eqref{eq:frictionless dynamic} strategy $\bar{\varphi}$ in the frictionless market and the forward equation~\eqref{eq:eqphi}, this process has dynamics 
\begin{align}\label{eq:dphidyn}
 d\Delta\varphi_t = \left(\left(G'\right)^{-1}\left(\frac{Y_t}{\lambda_t}\right) - \bar{b}_t\right)  dt - \bar{a}_tdW_t, 
 && \Delta\varphi_0 = \varphi_{0-} + \frac{\xi_0}{\sigma_0}-\frac{\mu_0}{\gamma\sigma_0^ 2}.
 \end{align}
Accordingly, plugging~\eqref{eq:strat} into~\eqref{eq:bwopt}, the backward process $Y$ therefore becomes
\begin{align}\label{eq:BSDEY}
dY_t =\gamma \sigma_t^2 \Delta \varphi_t dt + Z_t dW_t , &&{Y}_T=0. 
\end{align}

\begin{remark}

The analysis is focused on the optimal marginal transaction costs $Y$ instead of the optimal trading rates $\dot\varphi$, to avoid the potential extra requirement on the differentiability of $\left(G'\right)^{-1}$. Moreover, this choice of  backward component has linear dynamics, and absorb all the nonlinearity into the dynamic of the forward component, which makes our asymptotic analysis much easier and the numerical solution much more stable. 

We choose the deviation of the frictional position with respect to its frictionless counter party as the forward variable, since it exactly is the ``fast'' variable in the theoretical study in the single agent problem as in~\cite{bank.al.17,moreau.al.17,soner.touzi.13}. Empirically this choice performs the best in FBSDE solver. Notice that the drift $\bar{b}$ for the frictionless strategy  is usually neglected in the small-costs regime, and the frictionless volatility $\bar{a}$ makes the FBSDE system non-deterministic.  

\end{remark}

\section{Deep Learning-based Numerical Algorithms}\label{ML}

For practitioners, it is crucial to get the optimal trading strategy numerically, especially in a time-efficient manner. 
Therefore, 
 efficient numerical methods to solve the optimization problem~\eqref{general:regular goal} are in great need. 
With the development of GPUs and highly accurate machine learning models, the optimal hedging problems can be solved numerically by using learning-based algorithms.
There are two numerical approaches to do so: numerically solving the FBSDE system~\eqref{eq:dphidyn} - \eqref{eq:BSDEY} or directly targeting the goal functional~\eqref{general:regular goal}.    In this section,  we first present two popular deep learning-based numerical algorithms guided by these ideas: the FBSDE solver, and the Deep Hedging algorithm. Then, we highlight the advantages and disadvantages of these two methods in practice, especially when both of them fail to work. Finally, we introduce a transfer learning-based method, the {\em ST-Hedging} algorithm, to efficiently solve the optimal hedging problems when others fail to do so. 

A natural idea to start with is to solve the associated FBSDE system~\eqref{eq:dphidyn} - \eqref{eq:BSDEY} 
numerically. 
Since the dimension of the FBSDEs grows quickly with the number of assets and agents, and the boundary conditions are unclear, classical numerical methods such as finite differences fail to work. 
However, the learning-based algorithms can bypass the need to identify the correct boundary conditions and overcome the curse of dimensionality. 
The first learning-based introduced in Section~\ref{sub:FBSDE Solver} is the {FBSDE solver}, proposed in the spirit of the BSDE approach by Han, Jentzen, and E in~\cite{han.al.17}. 
The algorithm approximates the dependence of the backward components on the forward components by a deep neural network, and it solves the FBSDE through simulation and stochastic gradient descent methods. 
It can also handle higher-dimensional settings, e.g., random and time-varying transaction costs, returns, and volatility processes. 
In addition, the FBSDE solver can even solve equilibrium models as is used in~\cite{gonon2021asset}. FBSDE solver works very well when the time horizon is not too long, and the convergence of this method with respect to a special class of FBSDE systems is established in small time durations in~\cite{han2020convergence}. 
However, the FBSDE solver does not scale well when the trading horizon or  time discretization is large. 
In our example,  with calibrated market parameters from~\cite{gonon2021asset}, the FBSDE solver fails to converge if we consider a more than one trading year trading horizon (approximately 252 trading days). Furthermore, in some complicated cases, we cannot even characterize the optimal trading problem with a  system of FBSDEs. 

In order to overcome the difficulties mentioned above, we consider deep reinforcement learning framework that is known as Deep Hedging in Section~\ref{sub:utility}. 
The name for this type of algorithms comes from the groundbreaking {Deep Hedging} work of Buehler et.al.~\cite{buehler2019deep,buehler2019deepa}. 
Pioneered by~\cite{becker2019deep,han2016deep,hu2019deep, hure2018deep, min2021signatured, Moallemi2021RL, sun2020optimizing,nevmyvaka2006reinforcement,reppen2020bias,ruf2021hedging}, various reinforcement learning algorithms
 are implemented and perform extreme success in portfolio optimization problems with transaction costs. 
Reinforcement learning can even solve the Nash equilibria, see~\cite{casgrain2019deep} for details. 
The key idea is to  directly parametrize the optimal trading rate and optimize the discretized version of preference~\eqref{general:regular goal}. Then, through stochastic gradient descent based algorithms and back-propagation, the reinforcement algorithm updates the parameters of the networks until a (local) optimizer is found (see~\cite{sutton2018reinforcement} for a detailed introduction). However, Deep Hedging algorithms require a huge number of simulated sample paths and finer discretization of the trading time horizon, which makes the training time significantly longer than FBSDE solver. 

One major drawback for both algorithms is that their performances become suboptimal when the trading time horizon is long. Just as in the experiments illustrated in Section~\ref{experiment: quadratic} and Section~\ref{exp:3/2}, when the investment horizon is more than one trading year, FBSDE solver fails to converge. Although Deep Hedging still works, it suffers from overfitting and becomes difficult to tune. 
Moreover, when there are more than one stocks in the market with cross-sectional effects, the FBSDE solver and Deep Hedging both fail to converge even with intermediate trading horizons. In other words, these methods do not generalize well with the increase of the dimensions due to coupling. 

In parallel, as in~\cite{almgren.li.16,bayraktar.al.18, guasoni.weber.18,kallsen.muhlekarbe.17,moreau.al.17,shreve1994optimal,soner.touzi.13}, researchers have focused on the asymptotic optimal trading strategy in the small transaction costs limit. Especially, there are explicit asymptotic formulas available for the optimal trading strategies. 
A nature following-up question is if we can take advantage of both the advantages of the deep learning algorithms and the asymptotic approximations. 
With this idea, we propose our ST-Hedging algorithm that allow us to have the convenience of the explicit asymptotic approximation and the accuracy of the  deep learning algorithm.  
ST-Hedging reaches state-of-the-art performance comparing to Deep Hedging for long investment horizons as in Section~\ref{experiment: quadratic} and Section~\ref{exp:3/2}. In addition, it is scalable with respect to dimensions as in Section~\ref{exp:multi}. 

The organization of the rest of the section is as follows. First, we introduce the FBSDE solver in Section~\ref{sub:FBSDE Solver} and the Deep Hedging algorithm in Section~\ref{sub:utility}. 
In Section~\ref{sub:pasting},  we introduce the leading-order approximation formula and propose our ST-Hedging algorithm.


\subsection{FBSDE Solver}\label{sub:FBSDE Solver}

Let us describe the algorithm in more detail. 
The key observation is that the dynamic of the system 
~\eqref{eq:dphidyn} - \eqref{eq:BSDEY} is pinned down if the volatility $Z$ of the backward component $Y$ is specified. 
Indeed, let us fix a time partition $0=t_0 <  \ldots < t_{N} = T$, where $t_m = mT/N$ and~$\Delta t = T/N$.  Let $(\Delta W_{m})_{m=0}^N$ be iid normally distributed random variables with mean zero and variance $\Delta t$. 
At time $t_m$,  we have $\mu_{t_m}$, $\sigma_{t_m}$ and $\lambda_{t_m}$ from the market and the endowment volatility $\xi_{t_m}$, the drift $ \bar{b}_{t_m}$ and volatility $ \bar{a}_{t_m}$ for the frictionless strategy of the agent.
With $\left(\Delta W_{m},\Delta \varphi_{t_m} \right)$ as input for each time step, the discrete-time analogue of the forward update rule for the FBSDE
system~\eqref{eq:dphidyn} - \eqref{eq:BSDEY} becomes
\begin{align}
\Delta\varphi_{t_{m+1}} &= \Delta\varphi_{t_{m}}  + \left(\left(G'\right)^{-1}\left(\frac{Y_{t_m}}{\lambda_{t_m}}\right) - \bar{b}_{t_m}\right) \Delta t - \bar{a}_{t_m} \Delta W_m, 
\label{eq:varphi}\\
Y_{t_{m+1}} & = Y_{t_m} + \gamma \sigma_{t_m}^2 \Delta\varphi_{t_m} \Delta t + Z_{t_m} \Delta W_m.
\label{eq:Y}
\end{align}
Together with a guess for the initial values of the backward components, we can then simulate the system with a standard forward scheme and check whether it satisfies the terminal condition as $Y_T = 0$. 
Searching for the correct initial values is relatively straightforward. 
The more difficult challenge is parametrizing the ``controls'' $Z$ -- which are functions of time and the forward processes --  and updating the corresponding initial guesses until the terminal condition is matched sufficiently well. In~\cite{han.al.17}, Han, Jentzen, and E introduced the innovative way to parametrize $Z$. 
At each time point $t_m $, they parametrized $Z_{t_m}$ with a (shallow) network structure $F^{\theta_m}$, with time $t_m$ value of the forward components as the inputs. 
Together with the guess of the initial value of the backward component, which we denote as $Y^\theta_0$, 
 the system can then be simulated forward in time.  In this way, the shallow network structures are concatenated over time and become a deep neural network architecture, since the number of layers of the architecture and the number of parameters in the networks grows linearly in the number of discretization steps $N$. In other words, we input simulated Brownian paths into the deep network structure, and it outputs the terminal values of the backward components.

Our task is to update the parameters $ \{Y^\theta_0, \theta_m, m=0,\ldots,N\}$ until the terminal conditions $Y_T = 0$ are matched sufficiently well. 
This iterative update of the network parameters until $Y_T=0$ is the essence of \emph{supervised machine learning} tasks. State-of-the-art performance is achieved through back-propagation and stochastic gradient descent-type algorithms, see \cite{Goodfellow2016} for more details. 

\begin{remark}
Recall that there is no assumption on the dynamic of the market. Here, however, we need to mildly assume that we can simulate the return $\mu$ and volatilities $\sigma$,  the drift $\bar{b}$ and volatilities $\bar{a}$ for the frictionless position $\bar\varphi$ from~\eqref{eq:frictionless dynamic}, as well as the endowment volatilities $\xi$. 
\end{remark}

Now, we focus on the parametrization of $\{{Z}_{t_m}\}_{m=0}^N$ within a function class $\{ F^{{\theta}}: \theta\in\Theta \}$. A very popular class of functions in the machine learning community is the ``singular activation function Rectified Linear Unit'' (\emph{ReLU}):
$\text{\emph{ReLu}}(x) = \max\{x, 0\}$.
A popular class of approximation function $F$ is the convolution of linear functions with \emph{ReLu} activations. 
For the numerical experiments in Section~\ref{results}, at each time $t_m$, the $F^{{\theta_m}}$ we use is a neural network 
{with one hidden layer of 15 hidden units with batch normalization (BN), and that is 
\begin{align*}
F^{\theta_m} (x) =   w^2_{\theta_m} \left( \text{\emph{ReLu}} \left( BN \left(w^1_{\theta_m} x + b^1_{\theta_m} \right)\right)\right) + b^2_{\theta_m} .
\end{align*}
Recall that with the Brownian motion $W$ as a $1$-dim process, and the forward component $\varphi$ as a 1-dim vector, $Z$ is also a 1-dim vector. $w^1_{\theta_m} \in\R^{15\times 2}$ and $b^1_{\theta_m}\in\R^{15}$ are called the input \emph{weights} and \emph{biases} of the network, whereas $w^2_{\theta_m} \in\R^{1 \times 15}$ and $b^2_{\theta_m}\in\R^{1}$ are called the output \emph{weights} and \emph{biases} of the network. 
} We summarize the structure of FBSDE solver in Algorithm~\ref{algo:FBSDE Solver}: 
\medskip

\RestyleAlgo{boxruled}
\LinesNumbered
\begin{algorithm}[H]
  \caption{Training Procedure of FBSDE Solver~\label{algo:FBSDE Solver}}
   \KwData{
   {${\Delta \varphi}^{\theta}_{t_0}= \varphi_{0-} + \frac{\xi_0}{\sigma_0}-\frac{\mu_0}{\gamma\sigma_0^ 2}$,  $W_{t_0}=0$, $\gamma$, dynamic of processes $\mu$, $\sigma$, $\lambda$, $\xi$,  $\bar{b}$ and $\bar{a}$};
   }
 initialization of parameters $\{Y^\theta_0,  \theta_m, m=0,\ldots, N\}$,
$m=0,\ \ Y^\theta_{t_0} = Y^\theta_0$;\\
 \While{epoch $\leq$ \texttt{Epoch}}
 {
sample $\Delta W$,  $\texttt{batch\_size}\times N$  iid Gaussian random variables with variance $\Delta t$;\\
\While{$m<N$}{
$Z^{\theta}_{t_m} \quad\ \ = F^{\theta_m}(W_{t_m}, {\Delta \varphi}^{\theta}_{t_m})$;\\
$Y^\theta_{t_{m+1}}
\ \ \ =Y^\theta_{t_m} + 
\gamma \sigma_{t_m}^2 {\Delta \varphi}^{\theta}_{t_m} \Delta t 
+ Z^\theta_{t_m} \Delta W_m $ ; \\
${\Delta \varphi}^{\theta}_{t_{m+1}} ={\Delta \varphi}^{\theta}_{t_m} + \left(G'\right)^{-1}\left({Y^\theta_{t_m}}/{\lambda_{t_m}}\right)  \Delta t -\bar{b}_{t_m}\Delta t - \bar{a}_{t_m} \Delta W_m$; \\
$W_{t_{m+1}} \ \ = W_{t_m}+\Delta W_m$;\\
$m$++;\\
 }
Output ${Y}^{\theta}_T = {Y}^{\theta}_{t_N}$;\\
$\texttt{Loss} = \| Y^\theta_T\|^2/{\texttt{batch\_size}}$; \\
Calculate the gradient of \texttt{Loss} with respect to $\theta$;\\
Back propagate updates for $\{Y^\theta_0,  \theta_m, m=0,\ldots, N\}$ via \texttt{Adam};\\
epoch ++;
 }
\Return:  (local) optimizer $\theta^* = \{Y^{\theta^*}_0, \theta^*_m, m=0,\ldots, N\}$.\\
\end{algorithm}

As studied by~\cite{han2020convergence}, when the short time existence of the FBSDE system is available, the convergence rate of the FBSDE solver is guaranteed.
In reality, FBSDE system with short time existence does not always generalize to global existence. That is to say, if we only extend the time horizon with all the parameter remain the same, the FBSDE solver might not converge anymore.
Especially, since only the terminal value of the backward components are compared, the FBSDE solver fails to work efficiently with even intermediate trading horizons. 

\textcolor{blue}{
}

\subsection{Deep Hedging Networks}\label{sub:utility}
Let us fix a time partition $0=t_0 <  \ldots < t_{N} = T$, where $t_m = mT/N$ and~$\Delta t = T/N$. 
Under this partition, for a single simulation, the discretized version of the goal functional~\eqref{general:regular goal}, becomes 
\begin{align}\label{discretized goal}
J_T(\dot\varphi) = \frac{1}{N} \sum_{m=0}^N \left[\varphi_{t_m} \mu_{t_m} - \frac{\gamma}{2}\left(\sigma_{t_m} \varphi_{t_m} + \xi_{t_m} \right)^2- \lambda_{t_m} G\left(\dot\varphi_{t_m}\right) \right],
\end{align}
where $\varphi$ follows the discretized version of update process
$$
\varphi_{t_{m+1}} =\varphi_{0-} +  \sum_{k=0}^m \dot\varphi_{t_k} \Delta t= \varphi_{t_m} +\dot\varphi_{t_m} \Delta t. 
$$

At each time point $t_m$, we directly parametrize the trading strategy $\dot\varphi_{t_m}$ using a (comparatively shallow) network $F^{\theta_m}$. 
For the numerical experiments in Section~\ref{results}, the neural network we use has three hidden layers with 10, 15 and 10 hidden units, respectively. We also implement the batch normalization (BN) at each hidden layer and that is, for $l=1,2,3$:
\begin{align*}
F^{\theta_m} (x) =  w^4_{\theta_m} \left( F^{\theta_m}_3 \left( F^{\theta_m}_2 \left( F^{\theta_m}_1 \left( x \right)\right)\right)\right) + b^4_{\theta_m},
\quad
 F^{\theta_m}_l (x) = \text{\emph{ReLu}} \left( BN \left(w^l_{\theta_m} x + b^l_{\theta_m} \right)\right). 
\end{align*}
Recall that with time $t$ as a real number, the Brownian motion $W$ as a $1$-dim process, and the forward component $\varphi$ as a 1-dim vector, $\dot{\varphi}$ is also a 1-dim vector. $w^1_{\theta_m} \in\R^{10\times 3}$ and $b^1_{\theta_m}\in\R^{10}$ are called the input \emph{weights} and \emph{biases} of the network, and $w^2_{\theta_m} \in\R^{15\times 10}$, $b^2_{\theta_m}\in\R^{15}$, $w^3_{\theta_m} \in\R^{10\times 15}$ and $b^3_{\theta_m}\in\R^{10}$ are called the hidden \emph{weights} and \emph{biases}, whereas $w^4_{\theta_m} \in\R^{1 \times 10}$ and $b^4_{\theta_m}\in\R^{1}$ are called the output \emph{weights} and \emph{biases} of the network. 
We summarize the Deep Hedging algorithm in Algorithm~\ref{algo:utility} as follows:
\medskip

\RestyleAlgo{boxruled}
\LinesNumbered
\begin{algorithm}[H]
  \caption{Training Procedure of Deep Hedging Algorithm ~\label{algo:utility}}
   \KwData{$\varphi^{\theta}_{t_0} = \varphi_{0-}$, $W_{t_0}=0$, $\gamma$, 
   dynamic of processes   $\mu$, $\sigma$, $\lambda$ and $\xi$}
 initialization of parameters $\{ \theta_m, m=0,\ldots, N\}$, 
$m=0$;\\
 \While{epoch $\leq$ \texttt{Epoch}}
 {
sample $\Delta W$,  $\texttt{batch\_size}\times N\times d$  iid Gaussian random variables with variance $\Delta t$;\\
\While{$m< N$}{
$\dot{\varphi}^\theta_{t_{m}}\quad\ = F^{\theta_{m}}(t_m, W_{t_m},\varphi^\theta_{t_m})$;\\
$\varphi^\theta_{t_{m+1}}\ = \dot{\varphi}^\theta_{t_m} \Delta t + \varphi^\theta_{t_m}$;\\
$W_{t_{m+1}} = W_{t_m} + \Delta W_{m}$;\\
$m$++;
}
$\texttt{Loss} = 
- \sum_{m=0}^N 
         \left[\varphi^\theta_{t_m} \mu_{t_m} - \frac{\gamma}{2}\left(\sigma_{t_m} \varphi^\theta_{t_m} + \xi_{t_m} \right)^2 - \lambda\Lambda_{t_m} G\left(\dot\varphi^\theta_{t_m}\right)  \right]
$;
\\
Calculate the gradient of \texttt{Loss} with respect to $\theta$;\\
Back propagate updates for $\{\theta_m, m=0,\ldots, N\}$ via \texttt{Adam} (switch to \texttt{SGD} when fine tuning);\\
epoch ++;
 }
\Return:  (local) optimizer $\theta^* = \{ \theta^*_m, m=0,\ldots, N\}$
\end{algorithm}

\subsection{Stable-Transfer Learning Algorithm: ST-Hedging}\label{sub:pasting} 
The optimal hedging problem can be treated in two folds.  On the numerical aspect, as the development of deep learning-based algorithm, we can also solve the system through various different scheme, with  is discussed in details in Section~\ref{ML}. 
On the theoretical aspect,  one can study the \emph{asymptotic} limiting behavior as the trading horizon $T$ goes to $\infty$. Researchers always focus on small transaction costs for tractability as in~\cite{almgren.li.16,bayraktar.al.18, guasoni.weber.18,kallsen.muhlekarbe.17,moreau.al.17,shreve1994optimal,soner.touzi.13}. 
A natural question is, when the trading horizon $T$ is identified as long and whether it is related to  the smallness assumption of the transaction costs level $\lambda$. 
It turns out,  the ``magic'' quantity is $\sqrt{\lambda}/{T}$; namely, the smallness assumption on $\lambda$ should not only be an absolute quantity, but also a relative quantity, depending on whether $\sqrt{\lambda}/{T}$ is of higher order of 1. Even more surprisingly, $O(\sqrt{\lambda}/{T})$ is the approximation accuracy for the leading-order approximation~\eqref{leading order: general} under rather general condition (see the discussion in Appendix~\ref{sec:asymptotic}), that is independent of the choice $G$ of the transaction costs. 

More specifically, the leading-order asymptotic optimal strategy $\dot\varphi = \left( \dot\varphi_t\right)_{t\in[0,T]}$ for the frictional mean-variance preference~\eqref{general:regular goal} is (with a little abuse of notation) essentially given by: 
\begin{align}\label{leading order: general}
\dot\varphi_t = -\sign{\left(\varphi_t - \bar\varphi_t\right)}\times (G')^{-1}\left(\frac{g(\left|\varphi_t - \bar\varphi_t\right|;\gamma,\sigma_t,\bar{a}_t,\lambda_t)}{\lambda_t}\right) , 
\end{align}
where $g$ is the unique solution to the following ODE with proper ``mean-reverting'' requirement:
\begin{equation}
 (G')^{-1}\left(\frac{g(x;\gamma,\sigma,\bar{a},\lambda)}{\lambda}\right)g'(x;\gamma,\sigma,\bar{a},\lambda) + \frac{\bar{a}^2}{2} g''(x;\gamma,\sigma,\bar{a},\lambda)  = \gamma\sigma^2 x. \label{eqn:ergodic ODE}
\end{equation}
This leading-order formula~\eqref{leading order: general} and the associated ODE~\eqref{eqn:ergodic ODE} are proposed and studied formally in~\cite{bayraktar.al.18, caye2017trading,guasoni.weber.18,kallsen.muhlekarbe.17,soner.touzi.13}, and we refer the readers to  Appendix~\ref{sec:asymptotic} for details of the derivation and rigorous proof in a special case.
Notice that once we obtain the numerical solution of~\eqref{eqn:ergodic ODE}, the leading-order formula can be applied \emph{immediately}.

A nature following-up question is if we can take advantage of both the advantages of the deep learning algorithms and the leading-order approximations. 
More precisely, can we have the convenience of the leading-order approximation formula and the accuracy of the  deep learning algorithm near the very end of the terminal date. 
Based on the leading-order approximation accuracy $O(\sqrt{\lambda}/T)$, we answer the above question in affirmation. 
%

To fully utilize the advantages, we need to adjust the  deep learning-based  algorithm. 
While we can obtain good performance under the current setup for the Deep Hedging algorithm, the hedging policy learned by the deep neural network becomes more and more unstable for longer time horizons. This is because the  variance induced by the Brownian motion $W_t$ is proportional to the time $t$, hence the hedging positions $\varphi_t$ is going to blow up as $t$ increases. To point out, with more than one stocks, the coupling effects will further contribute to the instability of $\varphi_t$.  
Instead, inspired by the variable choice of the FBSDE solver, we use the ``fast variable'':  $\Delta \varphi_t = \varphi_t - \bar{\varphi}_t$. 
As the apriori study in~\cite{almgren.li.16,bayraktar.al.18, guasoni.weber.18,kallsen.muhlekarbe.17,moreau.al.17,shreve1994optimal,soner.touzi.13}, 
the ``fast variable'' $\Delta\varphi_t - \bar{\varphi}_t$ is a mean-reverting process that quickly oscillates around zero. 
In particular, 
the variance of $\Delta\varphi_t$ is upper bounded hence stabilized as time increases. As long as the deep learning algorithm on a short trading horizon is well-trained, then it can be easily adapted to arbitrarily long trading horizons without scaling up the errors. Numerical experiments have demonstrated the success of our variance reduction technique as in~Section~\ref{results}. This choice is reflected in the ``stable'' part of the ST-Hedging. 

For the ``transfer'' part of ST-Hedging, we leverage \emph{transfer learning} technique to automatically pick the optimal switching time from the approximation formulas to accurately learning of the optimal strategy via deep learning-based algorithm.
By an initialized $\kappa$, we initialize our starting location $t_M$ near the end of the trading horizon by the suggestion from the approximation accuracy $T-t_M = O(\sqrt{\lambda})$. 
Then, with semi-well trained network, we simulate the rest of hedging positions through the stable hedging strategy, and compare it with the asymptotic hedging strategy of the accumulated gain. Then, we select the optimal $M^*$ such that the stable hedging loss is less than the asymptotic optimal trading loss by the largest amount and convert it back to $\kappa = (T-t_{M^*})/\sqrt{\lambda}$. We repeat the process iteratively until $\kappa$ converges, which equivalently means the optimal switching time $t_{M^*}$ is obtained.
Notice that we alternatively learn the optimal trading strategy and the optimal switching time, and the output of ST-Hedging is the time $t_{M^*}$ to switch from the asymptotic strategy and the network parameters for the learnt trading strategy afterwards. 
 \\

\RestyleAlgo{boxruled}
\LinesNumbered
\begin{algorithm}[H]
  \caption{Training Procedure of ST-Hedging ~\label{algo:paste}}
   \KwData{
   {${ \varphi}^{(0)}_{t_0}= \varphi_{0-}$,  $W_{t_0}=0$, $\gamma$, $\kappa$, $\lambda$,  dynamic of processes $\mu$, $\sigma$, $\Lambda$, $\xi$,  $\bar{\varphi}$, $\bar{a}$};
   }

 \While{$\texttt{epoch} \leq \texttt{Epoch}$}
 {
sample $\Delta W$,  $\texttt{batch\_size}\times N$  iid Gaussian random variables with variance $\Delta t$;\\
\texttt{\# Calculation of asymptotic optimal trading strategy}\;
$m=0, \ \Delta\varphi_{t_0}^{(0)}=0$\;
\While{$m<N$}{
$\dot\varphi^{(0)}_{t_m} = -\sign\left(\Delta\varphi^{(0)}_{t_m} \right) (G')^{-1}\left({g\left(\left|\Delta\varphi^{(0)}_{t_m}\right|; \gamma, \sigma_{t_m},\bar{a}_{t_m},\lambda\Lambda_{t_m}\right)}/{\lambda\Lambda_{t_m}}\right)$\;
$\varphi^{(0)}_{t_{m+1}}
 =\varphi^{(0)}_{t_m}  + \dot\varphi^{(0)}_{t_m}
\Delta t;$  \\
$\Delta\varphi^{(0)}_{t_{m+1}} = \varphi^{(0)}_{t_{m+1}} - \bar\varphi_{t_{m+1}}$\;
$W_{t_{m+1}} \ \ = W_{t_m}+\Delta W_m$;\\
$m$++;\\
 }
 \texttt{\# Training of stable hedging strategy}\;
$M = \min\left\{m: T - t_{m} < \kappa \sqrt{\lambda} \right\}, m=M$, randomly initialize $\Delta \varphi_{t_M}^\theta$\;
\While{$m< N $}{
$\varphi^\theta_{t_m} = \Delta\varphi^\theta_{t_m} +\bar\varphi_{t_m}$;\\
$\dot{\varphi}^\theta_{t_{m}}\quad\ = F^{\theta_{m}}(t_m, \Delta\varphi^\theta_{t_m} )$;\\
$\Delta\varphi^\theta_{t_{m+1}}\ = \dot{\varphi}^\theta_{t_m} \Delta t + \varphi^\theta_{t_m}- \bar\varphi_{t_{m+1}}$;\\
$m$++;
}
$\texttt{Loss} = 
- \sum_{m=M}^{N } 
         \left[\varphi^\theta_{t_m} \mu_{t_m} - \frac{\gamma}{2}\left(\sigma_{t_m} \varphi^\theta_{t_m} + \xi_{t_m} \right)^2 - \lambda\Lambda_{t_m} G\left(\dot\varphi^\theta_{t_m}\right)  \right]
$;
\\
Calculate the gradient of \texttt{Loss} with respect to $\theta$;\\
Back propagate updates for $\theta$ via \texttt{Adam} (switch to \texttt{SGD} when fine tuning);\\

\texttt{epoch} ++;

\If{$\texttt{epoch} \ \%\ 1000 ==0$}{
\texttt{\# Training of transfer time}\;
$\kappa^* =\kappa$;
 
\While{$\kappa^*$ Not Converged}{
	$M = \min\left\{m: T - t_{m} < \kappa^* \sqrt{\lambda} \right\}$\;


$\texttt{Loss}_0 (M)=- \sum_{m=M}^{N} 
         \left[\varphi^{(0)}_{t_m} \mu_{t_m} - \frac{\gamma}{2}\left(\sigma_{t_m} \varphi^{(0)}_{t_m} + \xi_{t_m} \right)^2 - \lambda\Lambda_{t_m} G\left(\dot\varphi^{(0)}_{t_m}\right)  \right]
$;

$\texttt{Loss}_\theta(M) = - \sum_{m=M}^N 
         \left[\varphi^\theta_{t_m} \mu_{t_m} - \frac{\gamma}{2}\left(\sigma_{t_m} \varphi^\theta_{t_m} + \xi_{t_m} \right)^2 - \lambda\Lambda_{t_m} G\left(\dot\varphi^\theta_{t_m}\right)  \right]
$;

$M^* = \mbox{argmin}_M \left(\texttt{Loss}_\theta(M) - \texttt{Loss}_0(M)\right)
$\;
$\kappa^* = \left({T - t_{M^*}}\right)/{\sqrt{\lambda}}$\;
}
$\kappa=\kappa^*$\;
}
 }

\Return:  (local) optimizer  $M^*, \theta^* = \{ \theta^*_m, m=M^*,\ldots, N\}$.
\end{algorithm}

\section{Experiments and Comparison Results}\label{results}
In this section, the FBSDE solver in~\ref{sub:FBSDE Solver}, the Deep Hedging algorithm in~\ref{sub:utility} as well as the proposed ST-Hedging algorithm in~\ref{sub:pasting} are implemented. 
We  test them on various models which even includes more than one stocks and/or with stochastic liquidity risk, and we discuss and document their advantages and disadvantages in these empirical implementations in Section~\ref{sub: comparison}. 

To illustrate the performance of the FBSDE solver, the Deep Hedging and the ST-Hedging algorithm, we consider two calibrated Bachelier models without liquidity risk: a quadratic transaction costs model, and a general power costs model with elastic parameter $q = 3/2$. 
In all experiments, the parameters are from the calibration of real-world time series data in~\cite{gonon2021asset}.  The agent's risk aversion is set to be $\gamma =  1.66\times10^{-13} $, the total shares outstanding in the market is $s = 2.46\times 10^{11}$, and the  stock return is $\mu =0.072$ with the stock volatility being  $\sigma = 1.88 $ 
\footnote{In~\cite{gonon2021asset}, calibration is done for an equilibrium model of two agents where the asset returns depend on the risk aversion of both agents. In our single-agent model we just take the asset's parameters to provide a more realistic numerical analysis. 
 Although the models considered here are Bachelier, it can be generalized to Geometric Brownian motion models and the whole calculation and derivation follow if we switch our current analysis from shares positions to money in stock positions, with the unit of the transaction costs changed accordingly. }. 
 All models were trained using the free GPU service on Google Colab. 
The implementation details and codes can be found here:  \href{https://github.com/xf-shi/ML-for-Transaction-Costs/blob/main/README.md}{https://github.com/xf-shi/ML-for-Transaction-Costs/blob/main/README.md}. 


\subsection{Quadratic Transaction Costs $G(x) = x^2/2$.}\label{experiment: quadratic}

With $G(x)= x^2/2$, on top of the leading-order approximation, the FBSDE system~\eqref{eq:dphidyn} - \eqref{eq:BSDEY} becomes linear and the optimal trading strategy is given explicitly  by\footnote{ Details and derivations of the solution for the FBSDE system with constant quadratic costs~\eqref{eq:dphidyn} - \eqref{eq:BSDEY} can be found in Appendix~\ref{ss: exact}. }:
\begin{align}\label{optimal: quadratic}
\dot\varphi_t = - \sqrt{\frac{\gamma\sigma^2}{\lambda}}\tanh\left( \sqrt{\frac{\gamma\sigma^2}{\lambda}}(T-t) \right)\Delta \varphi_t. 
\end{align}

The remaining parameters are taken from the calibrations from  Section 5 in~\cite{gonon2021asset}, where under quadratic costs   the endowment volatility parameter $\xi= 2.19\times 10^{10} $ and  the liquidity level parameter is $\lambda = 1.08\times 10^{-10}$. 
With the trading horizon $T=10, 21, 42, 252, 2520$ trading days, the performance of both learning methods and the leading-order approximation are summarized in { Table~\ref{table: Single_q=2_TR=10} - \ref{table:Single_q=2_TR=2520}} and  illustrated in Figure~\ref{fig: Single_q=2_TR=10} - \ref{fig: Single_q=2_TR=2520}. To account for the large Monte Carlo error from the long trading horizons such as $T=2520$ trading days (which is approximately 10 trading years), we evaluated the performance of our models using 100 million sample paths.

When the trading horizon is relatively short ($T\leq 42$ trading days), the performance of the FBSDE solver is the closest to the ground truth path-wisely, see Figure~\ref{fig: Single_q=2_TR=10} - \ref{fig: Single_q=2_TR=42}. 
It means that when the FBSDE solver converges, it can learn the optimal trading strategy perfectly. 
It takes approximately 1-2 hours to complete the training with fine tuning. 
However, as is already experienced in~\cite{gonon2021asset}, when the trading horizon is longer than one trading year, the FBSDE solver fails to converge. 

On the other hand,    the Deep Hedging algorithm can still work well once the hyperparameters are appropriately tuned,  as illustrated in Figure~\ref{fig: Single_q=2_TR=252} - \ref{fig: Single_q=2_TR=2520} 
and {Table~\ref{table: Single_q=2_TR=252} - \ref{table:Single_q=2_TR=2520}}.  
However, we can see that under-fitting phenomena exist, and it becomes more severe as the trading horizon increases, meaning that the hyperparamters become harder to tune.
The training time of Deep Hedging varies with the length of the trading horizon. For short horizons, the training time is approximately 3-4 hours with fine tuning, which is longer than FBSDE solver. For longer trading horizons, the training and fine tuning time required can be as long as weeks, or even months, which is too long for practical use.

In addition, the performance of the leading-order asymptotic approximation~\eqref{leading order: general} 
becomes better and better as the time horizon becomes longer. 
From the mean squared error of $\dot \varphi_T/s$ in the last column of  {Table~\ref{table: Single_q=2_TR=10} - \ref{table:Single_q=2_TR=2520}} we can see that shortly around the terminal time, the leading-order approximation diverges substantially from 0. This observation is consistent with the figures of $\dot{\varphi}$ and $\varphi$ as in~Figure~\ref{fig: Single_q=2_TR=10} - \ref{fig: Single_q=2_TR=2520}.  
 In other words, we are giving up the performance at the very end of the trading horizon to get an overall near-optimal performance without spending hours or days on tuning the hyperparameters in the machine learning algorithms. In fact, when the trading horizon is as large as $T=2520$ trading days, leading-order approximation achieves the closest expected utility to the ground truth, while Deep Hedging is hard to train as the variance of the parametrized position $\varphi_t$ becomes very large with long trading horizon. 
 
Finally, ST-Hedging leverages the mean-reverting fast variables $\Delta \varphi_t$ to stabilize the variance of the parametrized variable. Therefore, comparing to Deep Hedging, it is scalable to arbitrarily large trading horizons. Moreover, ST-Hedging does not only achieve near optimal performance across all trading horizons, but also requires a much shorter training time to convergence. 
To wit, when the trading horizon is short, ST-Hedging, Deep Hedging, and FBSDE solver can achieve the same state-of-the-art performance which learn exactly the ground truth (see Table~\ref{table: Single_q=2_TR=10} - \ref{table: Single_q=2_TR=42}), while the leading-order approximation shows a clear gap, as in Figure~\ref{fig: Single_q=2_TR=10} - \ref{fig: Single_q=2_TR=42}. When the trading horizon gets larger, the ST-Hedging still maintains the state-of-art performance from the beginning to the end. Notice that in these long trading horizons, FBSDE solver can no longer converge. Deep Hedging shows a near optimal performance with fluctuation around the ground truth, but the training time increases significantly. While the leading-order approximation can achieve a near optimal utility, it deviates from the ground truth around the terminal time (see Table~\ref{table: Single_q=2_TR=252} -  \ref{table:Single_q=2_TR=2520}).  For all trading horizons, ST-Hedging requires less fine tuning compared to Deep Hedging, and the training can be finished within 2 hours. 

\begin{figure}[H]
    \centering
    \includegraphics[width=0.8\linewidth]{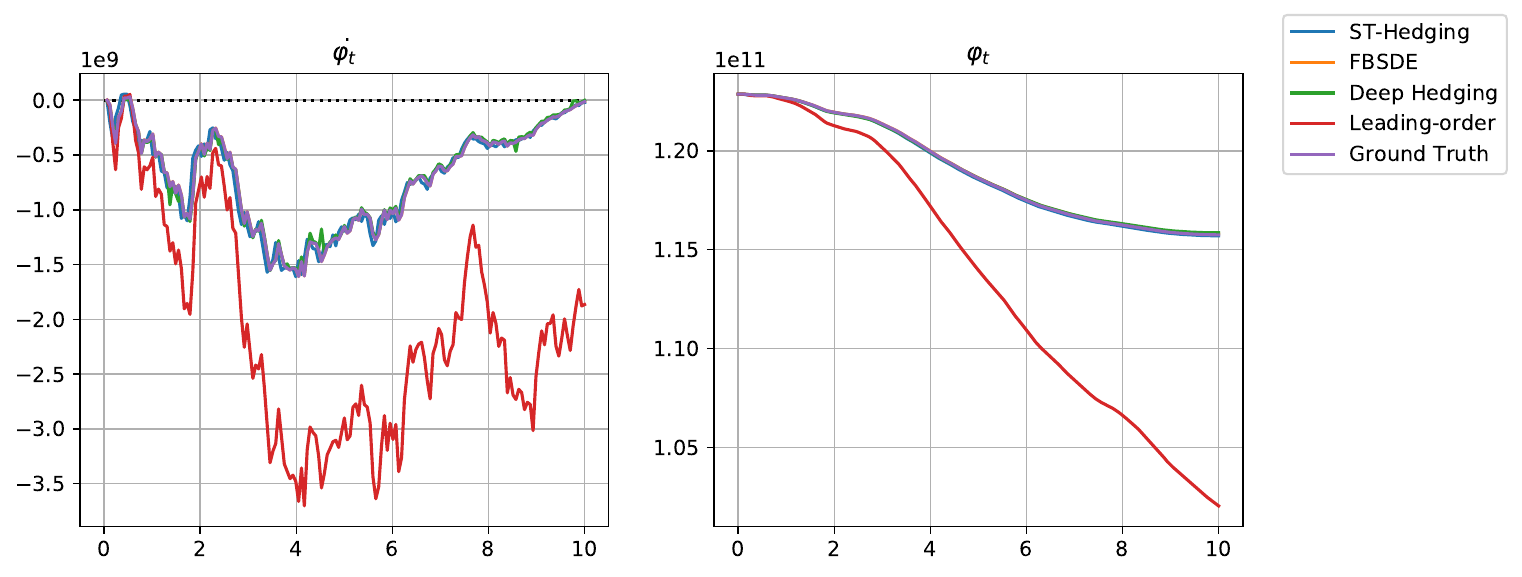}
 \caption{Optimal trading rates $\dot{\varphi}$ (left panel) and optimal positions $\varphi$ (right panel) for trading horizon $T = 10$ days with calibrated parameters under quadratic costs.}
    \label{fig: Single_q=2_TR=10}
\end{figure}
\begin{table}[H]
\centering
{\footnotesize{
\begin{tabular}{||c c c||} 
 \hline
 Method &  $J_T (\dot\varphi) \pm \mathrm{std}$ &$\E[|\dot\varphi_T|^2/s^2]$\\ [0.5ex] 
 \hline\hline  
ST-Hedging &$4.25\times 10^{9} \pm 1.51\times10^{9}$ & $2.64\times10^{-9}$  \\
Deep Hedging &$4.25\times 10^{9} \pm 1.52\times10^{9}$ & $2.57\times10^{-10}$  \\
 FBSDE Solver &$4.25\times10^{9}\pm 1.52\times10^{9}$ & $2.00 \times10^{-9}$ \\
 Leading-Order Approximation &$4.18\times10^{9}\pm1.53\times10^{9}$ & $6.33\times10^{-5}$ \\
 Ground Truth &$4.25\times10^{9} \pm 1.52\times10^{9}$ & $0.0$ \\ [0.5ex] 
 \hline
\end{tabular}
}}
\caption{Expectation and standard deviation of preference $J_T$, and mean squared error of $\dot \varphi_T/s$ for trading horizon $T = 10$ days with calibrated parameters under quadratic costs.}
\label{table: Single_q=2_TR=10}
\end{table}
\vspace{-15pt}
\begin{figure}[H]
    \centering
    \includegraphics[width=0.8\linewidth]{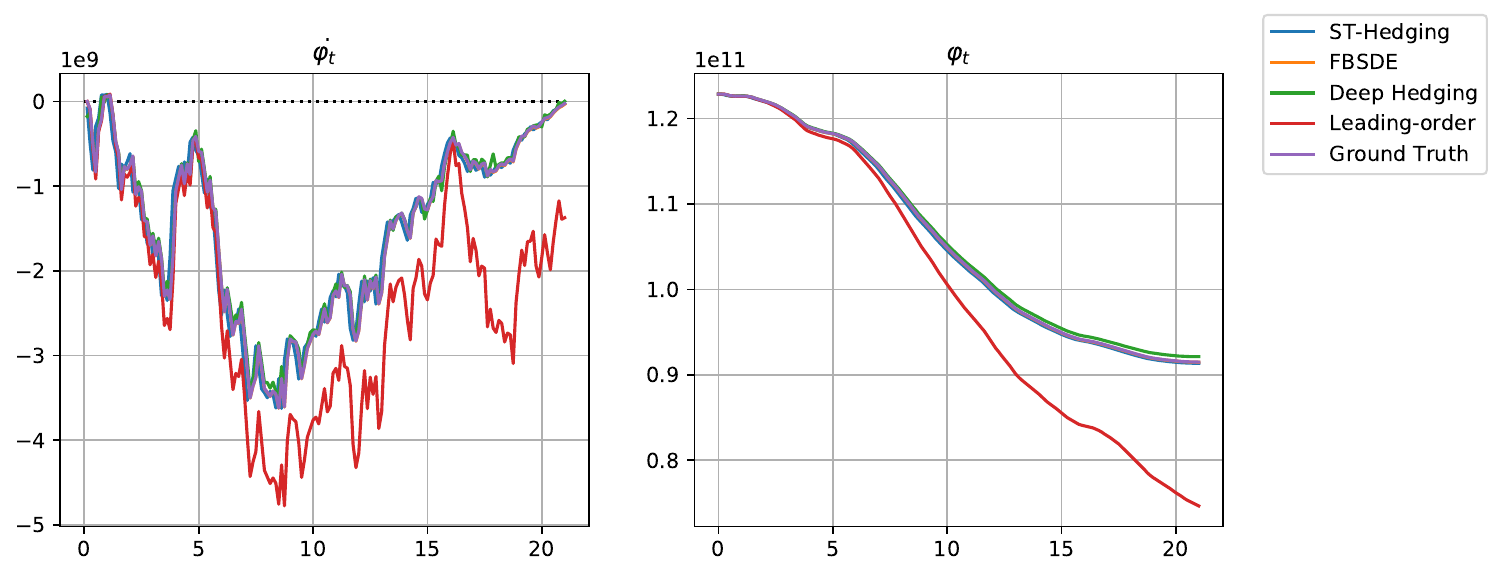}
    \caption{Optimal trading rates $\dot{\varphi}$ (left panel) and optimal positions $\varphi$ (right panel) for trading horizon $T = 21$ days with calibrated parameters under quadratic costs.}
    \label{fig: Single_q=2_TR=21}
\end{figure}

\begin{table}[H]
\centering
{\footnotesize{
\begin{tabular}{||c  c c||} 
 \hline
 Method &  $J_T (\dot\varphi) \pm \mathrm{std}$ &$\E[|\dot\varphi_T|^2/s^2]$\\ [0.5ex] 
 \hline\hline  
ST-Hedging &$4.13\times 10^{9} \pm 2.19 \times10^{9}$ & $1.26\times10^{-8}$  \\
Deep Hedging & $4.13\times 10^{9} \pm 2.20\times10^{9}$ & $3.62\times 10^{-9}$  \\
 FBSDE Solver & $4.13\times 10^{9} \pm 2.20\times10^{9}$ & $1.35\times10^{-8}$ \\
 Leading-order Approximation  & $4.06\times 10^{9}\pm 2.21\times 10^{9}$ & $7.89\times 10^{-5}$ \\
 Ground Truth  & $4.13\times 10^{9}\pm 2.20\times 10^{9}$ & $0.0$ \\ [0.5ex] 
 \hline
\end{tabular}
}}
\caption{Expectation and standard deviation of preference $J_T$, and mean squared error of $\dot \varphi_T/s$ for trading horizon $T = 21$ days with calibrated parameters under quadratic costs.}
\label{table: Single_q=2_TR=21}
\end{table}

\begin{figure}[H]
    \centering
    \includegraphics[width=0.8\linewidth]{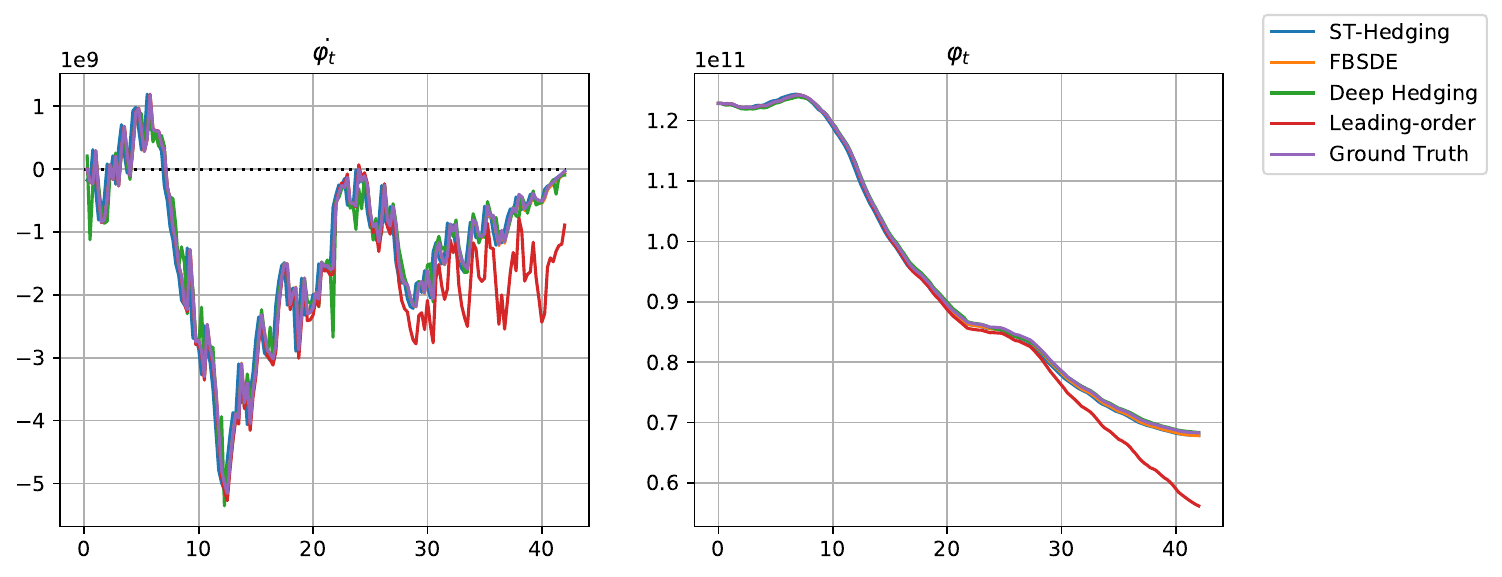}
    \caption{Optimal trading rates $\dot{\varphi}$ (left panel) and optimal positions $\varphi$ (right panel) for trading horizon $T = 42$ days with calibrated parameters under quadratic costs.}
    \label{fig: Single_q=2_TR=42}
\end{figure}
\begin{table}[H]
\centering
{\footnotesize{
\begin{tabular}{|| c c c||} 
 \hline
 Method &  $J_T (\dot\varphi) \pm \mathrm{std}$ &$\E[|\dot\varphi_T|^2/s^2]$\\ [0.5ex] 
 \hline\hline  
ST-Hedging &$3.97 \times 10^{9} \pm 3.22 \times10^{9}$ & $5.67 \times10^{-8}$  \\
Deep Hedging &$3.96\times10^{9} \pm 3.24\times10^{9}$ & $2.16\times10^{-7}$   \\
 FBSDE Solver &  $3.97\times10^{9} \pm  3.24\times10^{9}$ & $1.04\times10^{-7}$  \\
 Leading-order Approximation & $3.93\times10^{9} \pm 3.24\times10^{9}$ & $	8.56\times10^{-5}$  \\
 Ground Truth & $3.97\times10^{9} \pm  3.24\times10^{9}$ & $0.0$  \\ [0.5ex] 
 \hline
\end{tabular}
}}
\caption{Expectation and standard deviation of preference $J_T$, and mean squared error of $\dot \varphi_T/s$ for trading horizon $T = 42$ days with calibrated parameters under quadratic costs.}
\label{table: Single_q=2_TR=42}
\end{table}

\begin{figure}[H]
    \centering
    \includegraphics[width=0.8\linewidth]{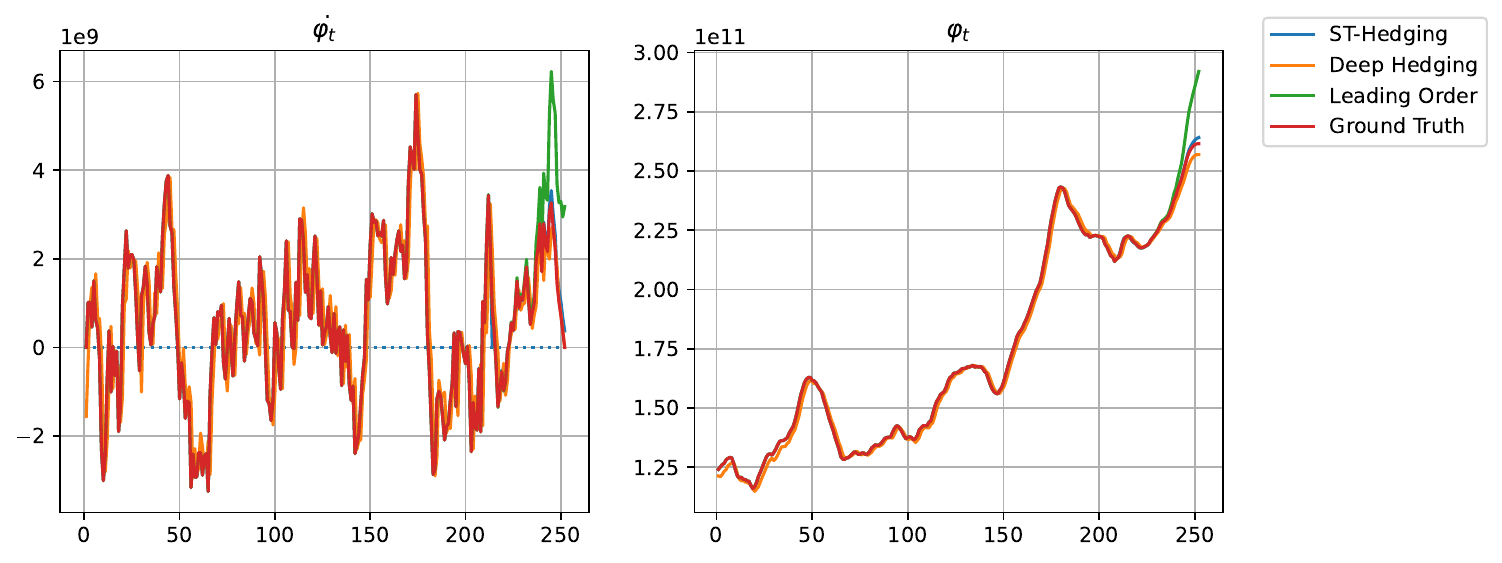}
    \caption{Optimal trading rates $\dot{\varphi}$ (left panel) and optimal positions $\varphi$ (right panel) for trading horizon $T = 252$ days with calibrated parameters under quadratic costs.}
    \label{fig: Single_q=2_TR=252}
\end{figure}
\begin{table}[H]
\centering
{\footnotesize{
\begin{tabular}{||c c c||} 
 \hline
 Method &  $J_T (\dot\varphi) \pm \mathrm{std}$ &$\E[|\dot\varphi_T|^2/s^2]$\\ [0.5ex] 
 \hline\hline  
ST-Hedging &	$3.91\times10^{9}\pm 8.61\times 10^{9}$ & $8.46 \times10^{-7}$   \\
Deep Hedging &	$3.87\times10^{9}\pm 8.59\times 10^{9}$ & $7.00\times10^{-8}$   \\
 FBSDE Solver & NaN & NaN  \\
 Leading-order Approximation & $3.91\times10^{9} \pm 8.61\times 10^{9}$ & $8.63\times10^{-5}$  \\
 Ground Truth & $3.91\times 10^{9} \pm  8.61\times 10^{9}$ & $0.0$  \\ [0.5ex] 
 \hline
\end{tabular}
}}
\caption{Expectation and standard deviation of preference $J_T$, and mean squared error of $\dot \varphi_T/s$ for trading horizon $T = 252$ days with calibrated parameters under quadratic costs. }
\label{table: Single_q=2_TR=252}
\end{table}

\begin{figure}[H]
    \centering
    \includegraphics[width=0.8\linewidth]{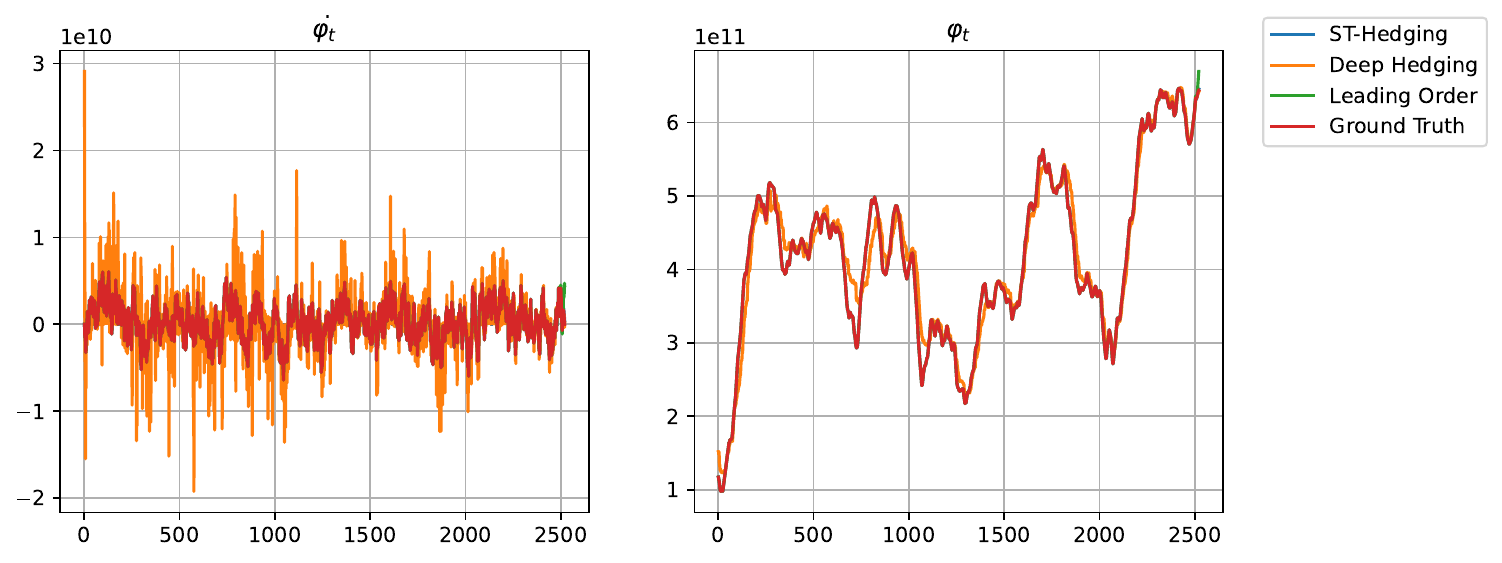}
    \caption{Optimal trading rates $\dot{\varphi}$ (left panel) and optimal positions $\varphi$ (right panel) for trading horizon $T = 2520$ days with calibrated parameters under quadratic costs.}
    \label{fig: Single_q=2_TR=2520}
\end{figure}
\begin{table}[H]
\centering
{\footnotesize{
\begin{tabular}{||c c c||} 
 \hline
 Method &  $J_T (\dot\varphi) \pm \mathrm{std}$ &$\E[|\dot\varphi_T|^2/s^2]$\\ [0.5ex] 
 \hline\hline  
ST-Hedging &	$3.87\times10^{9}\pm 2.47\times 10^{10}$ & $8.46 \times10^{-7}$   \\
Deep Hedging &	$3.32\times10^{9}\pm 2.47\times 10^{10}$ & $2.34\times10^{-7}$   \\
 FBSDE Solver & NaN & NaN  \\
 Leading-order Approximation & $3.87\times10^{9} \pm 2.47\times 10^{10}$ & $8.63\times10^{-5}$  \\
 Ground Truth & $3.87\times 10^{9} \pm  2.47\times 10^{10}$ & $0.0$  \\ [0.5ex] 
 \hline
\end{tabular}
}}
\caption{Expectation and standard deviation of preference $J_T$, and mean squared error of $\dot \varphi_T/s$ for trading horizon $T = 2520$ days with calibrated parameters under quadratic costs.}
\label{table:Single_q=2_TR=2520}
\end{table}

\subsection{Power Transaction Costs $G(x) =|x|^q/q $ with q=${3}/{2}$. }\label{exp:3/2}
The analogous experiments for power costs with $q = 3/2$  are also done with 
 calibrated parameters from Section 5 in~\cite{gonon2021asset}, where the endowment volatility is $\xi_{1.5} = 2.33\times10^{10}$ and liquidity level  $\lambda_{1.5} = 5.22\times10^{-6}$. 
In the general superlinear power transaction costs case, i.e. $G(x) = |x|^q/q$ with $q\in(1,2)$, the ground truth is no longer available in closed form. Nevertheless, we can still compare the numerical results from the machine learning algorithms and our leading-order asymptotic results. The models are evaluated using 100 million sample paths in order to account for the large Monte Carlo error from long trading horizons.

With the dynamics of the system being nonlinear,  the FBSDE solver converges on even shorter trading horizons, i.e., $T<21$ trading days, shown  in Figure~\ref{fig: Single_q=1.5_TR=10} - \ref{fig: Single_q=1.5_TR=21} and Table~\ref{table: Single_q=1.5_TR=10} - \ref{table: Single_q=1.5_TR=21}. When the trading horizon is large, i.e., $T \geq 42$ days, the FBSDE solver fails to converge.
The Deep Hedging algorithm obtains similar results as the FBSDE solver with short time horizons, as shown in Figure~\ref{fig: Single_q=1.5_TR=10} - \ref{fig: Single_q=1.5_TR=21}  and Table~\ref{table: Single_q=1.5_TR=10} - \ref{table: Single_q=1.5_TR=21}. 
With longer time horizons, the Deep Hedging algorithm can still work when the FBSDE solver fails to converge, see Figure~\ref{fig: Single_q=1.5_TR=42} - \ref{fig: Single_q=1.5_TR=2520} 
and Table~\ref{table: Single_q=1.5_TR=42} - \ref{table: Single_q=1.5_TR=2520}. However, when the trading horizon is as long as 10 trading  years, Deep Hedging becomes suboptimal because of the large variance of the position $\varphi_t$. 
Similar as in the quadratic costs case, with the trading horizon increasing, the performance of the leading-order approximation is getting closer to the optimal results learnt by the machine learning algorithm, which justifies the approximation result of the leading-order approximation~\eqref{leading order: general} empirically. 
ST-Hedging consistently achieves the best performance utility across all trading horizons (Table~\ref{table: Single_q=1.5_TR=10} - \ref{table: Single_q=1.5_TR=2520}). In short trading horizons, it achieves similar results as deep hedging (Figure~\ref{fig: Single_q=1.5_TR=10} - \ref{fig: Single_q=1.5_TR=42}). As the trading horizon gets larger, the noise of Deep Hedging becomes more and more severe, while ST-Hedging is still stable (Figure~\ref{fig: Single_q=1.5_TR=252} - \ref{fig: Single_q=1.5_TR=2520}).
The training and fine tuning time for these algorithms are similar to the quadratic costs case, with ST-Hedging requires significantly less training time and effort for tuning than the other methods.  

\begin{figure}[H]
    \centering
    \includegraphics[width=0.8\linewidth]{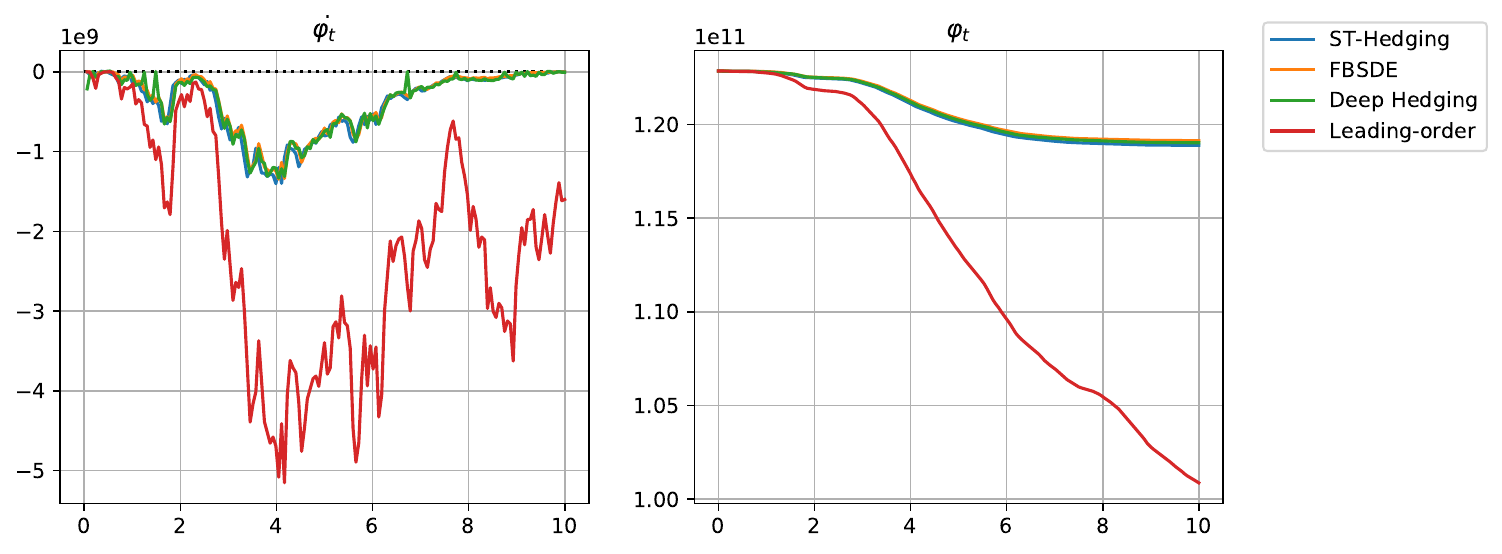}
    \caption{Optimal trading rates $\dot{\varphi}$ (left panel) and optimal positions $\varphi$ (right panel) for trading horizon $T = 10$ days with calibrated parameters under $q=3/2$ costs.}

    \label{fig: Single_q=1.5_TR=10}
\end{figure}
\begin{table}[H]
\centering
{\footnotesize{
\begin{tabular}{||c c c||} 
 \hline
Method &  $J_T (\dot\varphi) \pm \mathrm{std}$ &$\E[|\dot\varphi_T|^2/s^2]$\\ [0.5ex] 
 \hline\hline  
ST-Hedging &$4.22\times 10^{9} \pm 1.62 \times10^{9}$ & $3.62\times10^{-10}$  \\
Deep Hedging &	$4.22\times10^{9}\pm 1.62\times10^{9}$ & $3.22\times10^{-10}$  \\
 FBSDE Solver & $4.22\times10^9\pm 1.62\times10^{9}$ & $2.02\times 10^{-10}$  \\
 Leading-order Approximation & $4.11\times10^{9}\pm 1.54\times10^{9}$ & $8.76 \times10^{-5}$  \\ [0.5ex] 
 \hline
\end{tabular}
}}
\caption{Expectation and standard deviation of preference $J_T$, and mean squared error of $\dot \varphi_T/s$ for trading horizon $T = 10$ days with calibrated parameters under $q=3/2$ costs.}
\label{table: Single_q=1.5_TR=10}
\end{table}

\begin{figure}[H]
    \centering
    \includegraphics[width=0.8\linewidth]{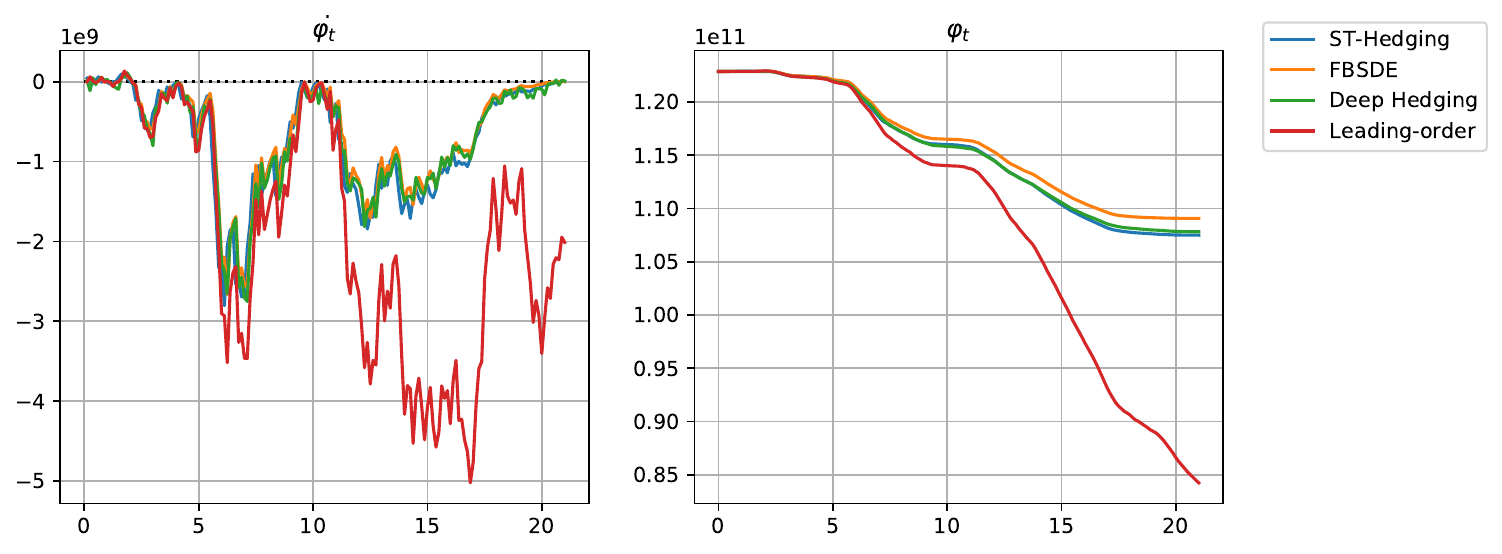}
    \caption{Optimal trading rates $\dot{\varphi}$ (left panel) and optimal positions $\varphi$ (right panel) for trading horizon $T = 21$ days with calibrated parameters under $q=3/2$ costs.}
    \label{fig: Single_q=1.5_TR=21}
\end{figure}
\begin{table}[H]
\centering
{\footnotesize{
\begin{tabular}{||c c c ||} 
 \hline
 Method &  $J_T (\dot\varphi) \pm \mathrm{std}$ &$\E[|\dot\varphi_T|^2/s^2]$\\ [0.5ex] 
 \hline\hline  
ST-Hedging &$4.02\times 10^{9} \pm 2.40 \times10^{9}$ & $1.34\times10^{-10}$  \\
Deep Hedging & $4.02\times 10^9\pm 2.42\times 10^9$ & $1.68\times10^{-9}$   \\
 FBSDE Solver &   $4.02\times 10^9 \pm  2.42\times 10^9$ & $4.55\times10^{-9}$  \\
 Leading-order Approximation & $3.93\times 10^9\pm 2.42\times 10^9$  & $1.10\times10^{-4}$  \\ [0.5ex] 
 \hline
\end{tabular}
}}
\caption{Expectation and standard deviation of preference $J_T$, and mean squared error of $\dot \varphi_T/s$ for trading horizon $T = 21$ days with calibrated parameters under $q=3/2$ costs.}
\label{table: Single_q=1.5_TR=21}
\end{table}

\begin{figure}[H]
    \centering
    \includegraphics[width=0.8\linewidth]{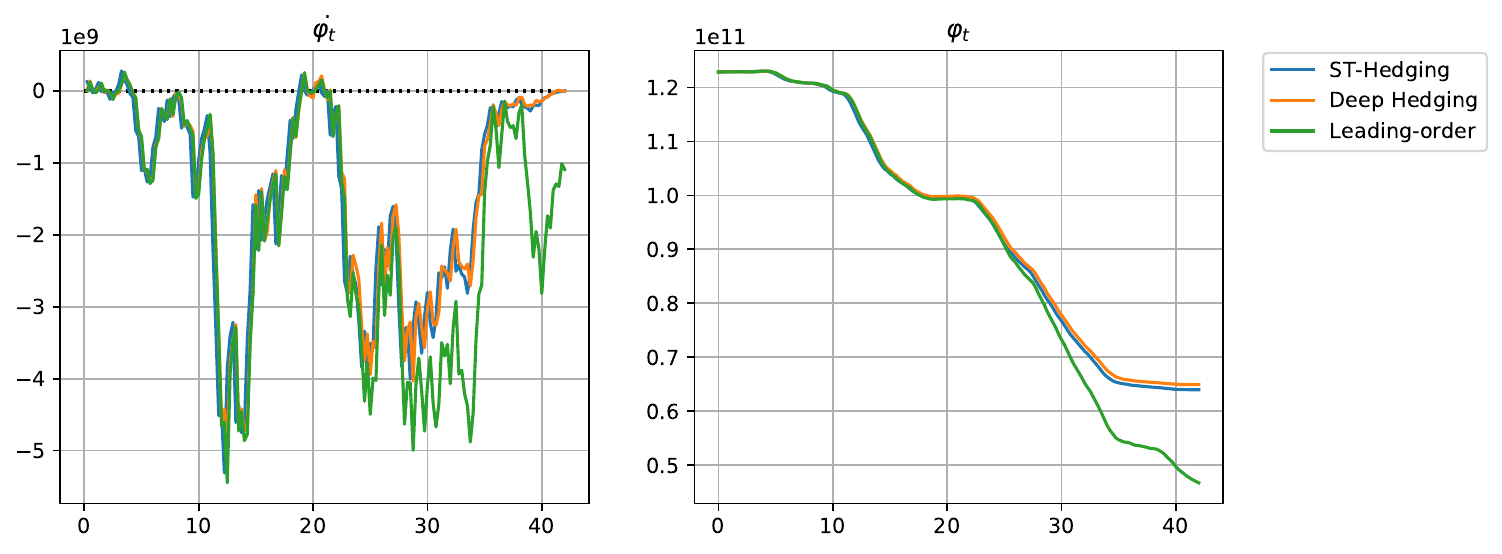}
    \caption{Optimal trading rates $\dot{\varphi}$ (left panel) and optimal positions $\varphi$ (right panel) for trading horizon $T = 42$ days with calibrated parameters under $q=3/2$ costs.}
    \label{fig: Single_q=1.5_TR=42}
\end{figure}
\begin{table}[H]
\centering
{\footnotesize{
\begin{tabular}{||c c c||} 
 \hline
 Method &  $J_T (\dot\varphi) \pm \mathrm{std}$ &$\E[|\dot\varphi_T|^2/s^2]$\\ [0.5ex] 
 \hline\hline  
ST-Hedging &$3.89 \times 10^{9} \pm 3.37 \times10^{9}$ & $1.36\times10^{-8}$  \\
Deep Hedging &	$3.89\times 10^9\pm 3.39\times 10^9$ & $1.53\times10^{-8}$  \\
 FBSDE Solver & NaN &NaN \\
 Leading-order Approximation & $3.84\times 10^9 \pm 3.39\times 10^9$ & $1.11\times10^{-4}$  \\ [0.5ex] 
 \hline
\end{tabular}
}}
\caption{Expectation and standard deviation of preference $J_T$, and mean squared error of $\dot \varphi_T/s$ for trading horizon $T = 42$ days with calibrated parameters under $q=3/2$ costs.}

\label{table: Single_q=1.5_TR=42}
\end{table}

\begin{figure}[H]
    \centering
    \includegraphics[width=0.8\linewidth]{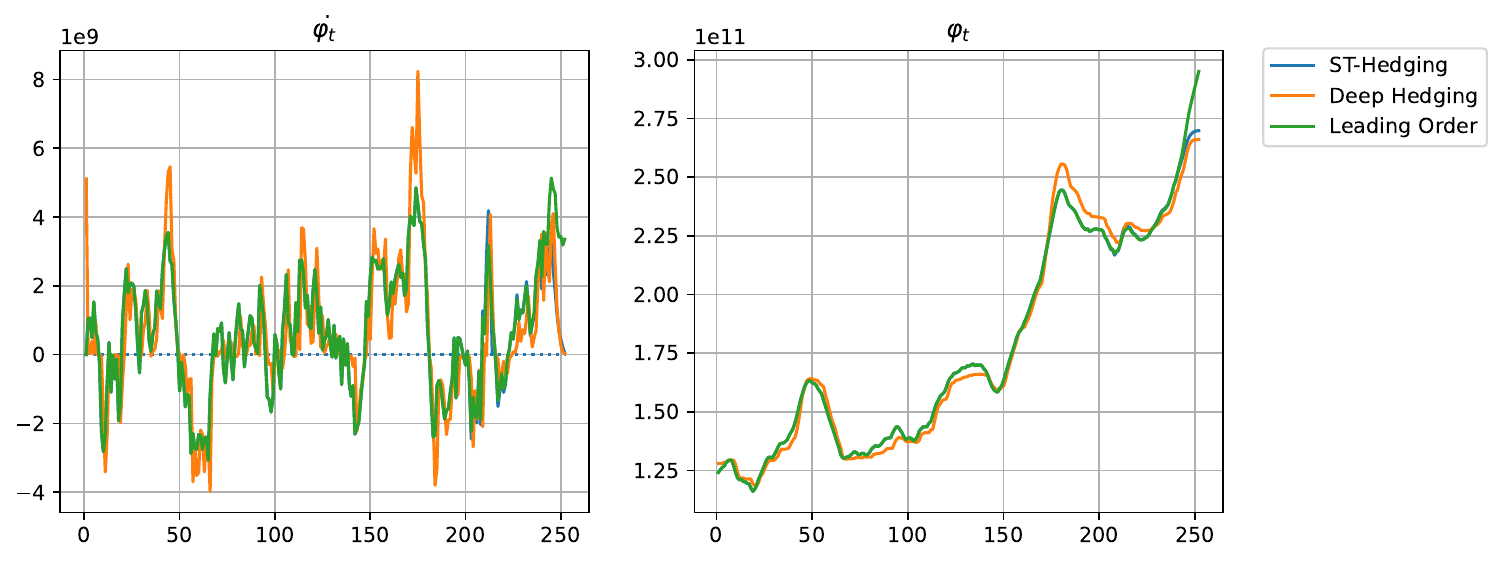}
    \caption{Optimal trading rates $\dot{\varphi}$ (left panel) and optimal positions $\varphi$ (right panel) for trading horizon $T = 252$ days with calibrated parameters under $q=3/2$ costs.}
    \label{fig: Single_q=1.5_TR=252}
\end{figure}
\begin{table}[H]
\centering
{\footnotesize{
\begin{tabular}{||c c  c||} 
 \hline
Method &  $J_T (\dot\varphi) \pm \mathrm{std}$ &$\E[|\dot\varphi_T|^2/s^2]$\\ [0.5ex] 
 \hline\hline  
ST-Hedging &	$3.84\times10^{9}\pm 9.18\times 10^{9}$ & $1.81 \times10^{-8}$   \\
Deep Hedging &	$3.84\times 10^9 \pm 9.18\times 10^9$ & $9.27\times 10^{-9}$   \\
 FBSDE Solver & NaN & NaN \\
 Leading-order Approximation & $3.84\times 10^9 \pm 9.18\times 10^9$  & $7.93\times 10^{-5}$  \\ [0.5ex] 
 \hline
\end{tabular}
}}
\caption{Expectation and standard deviation of preference $J_T$, and mean squared error of $\dot \varphi_T/s$ for trading horizon $T = 252$ days with calibrated parameters under $q=3/2$ costs.}
\label{table: Single_q=1.5_TR=252}
\end{table}
%
\begin{figure}[H]
    \centering
    \includegraphics[width=0.8\linewidth]{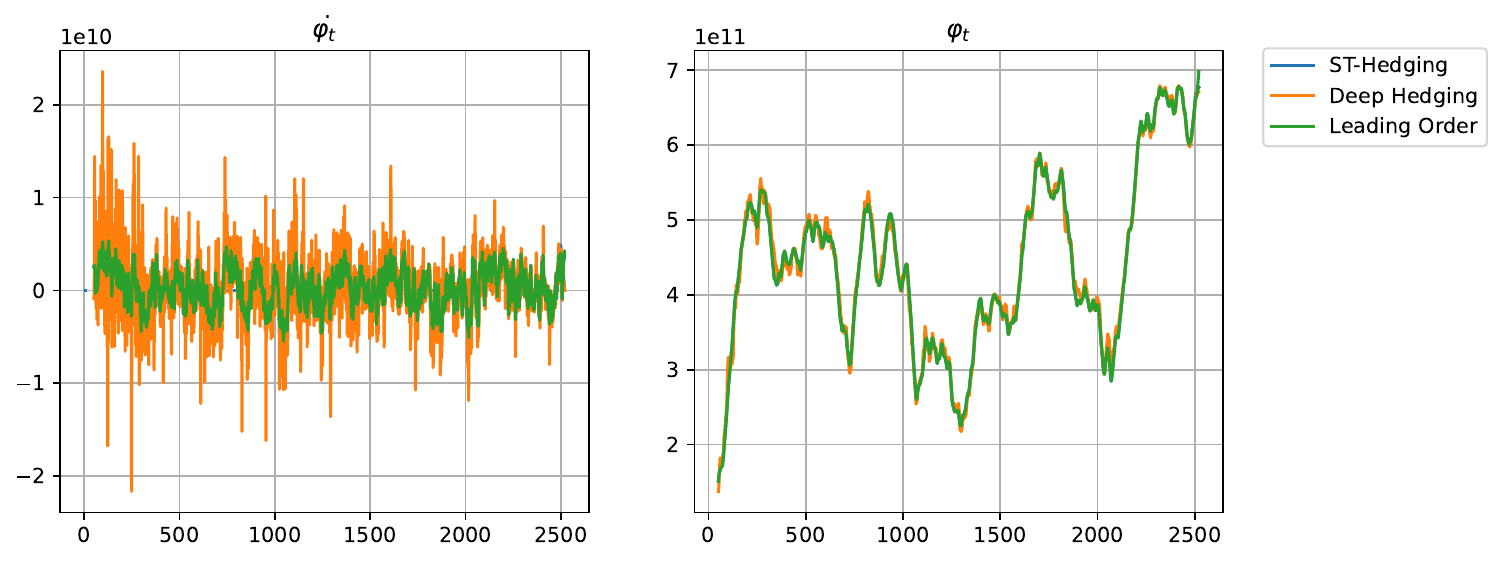}
    \caption{Optimal trading rates $\dot{\varphi}$ (left panel) and optimal positions $\varphi$ (right panel) for trading horizon $T = 2520$ days with calibrated parameters under $q=3/2$ costs.}
    \label{fig: Single_q=1.5_TR=2520}
\end{figure}
\begin{table}[H]
\centering
{\footnotesize{
\begin{tabular}{||c c  c||} 
 \hline
Method &  $J_T (\dot\varphi) \pm \mathrm{std}$ &$\E[|\dot\varphi_T|^2/s^2]$\\ [0.5ex] 
 \hline\hline  
ST-Hedging & $3.78\times 10^{9} \pm 2.59\times 10^{10}$  & $1.89\times 10^{-8}$  \\
Deep Hedging &	$3.59\times 10^{9} \pm 2.59\times 10^{10}$ & $8.38\times 10^{-8}$   \\
 FBSDE Solver & NaN & NaN \\
 Leading-order Approximation & $3.78\times 10^{9} \pm 2.59\times 10^{10}$  & $8.24\times 10^{-5}$  \\ [0.5ex] 
 \hline
\end{tabular}
}}
\caption{Expectation and standard deviation of preference $J_T$, and mean squared error of $\dot \varphi_T/s$ for trading horizon $T = 2520$ days with calibrated parameters under $q=3/2$ costs.}
\label{table: Single_q=1.5_TR=2520}
\end{table}

\subsection{Multi-asset Example via ST-Hedging}\label{exp:multi}

To illustrate the scalability of our ST-Hedging algorithm, we consider a model with three risky assets in the market with \emph{cross sectional effect}. In this experiment, the risky assets are driven by a 3-dim Brownian motion, and their corresponding annualized arithmetic variance is 
$$
\Sigma = \begin{pmatrix}
72.00 & 71.49 & 54.80\\
71.49 & 85.42 & 65.86\\
54.80 & 65.86 & 56.84\\
\end{pmatrix}. 
$$
The three stocks have outstanding shares are $1.15\times 10^{10}$, $3.2\times 10^9$ and $2.3\times 10^9$, respectively. The annualized arithmetic returns for the three stocks are set to be $2.99$, $3.71$, and $3.55$. 
Further, the risk aversion is set to be  $\gamma = 7.424 \times 10^{-13}$, and the endowment volatility matrix is set to be
$$\xi = \begin{pmatrix}
-2.07  & 1.91 & 0.64 \\
1.91  & -1.77 & -0.59 \\
0.64  & -0.59 & -0.20 \\
\end{pmatrix}\times 10^{9}.
$$
The transaction costs parameters are set as $1.269\times 10^{-9}$, $1.354\times 10^{-9}$, and $1.595\times 10^{-9}$. Finally, the trading time horizon is set to be $T = 2520$ trading days, which is 10 trading years, and the switching threshold we choose is $100$ days before maturity. 

The performance of the ST-Hedging algorithm and the comparisons to the leading-order formula and the ground truth are shown in Figure~\ref{fig: Multi_q=2_TR=2520} and Table~\ref{table: Multi_q=2_TR=2520_pasting}.
We can easily see that the overall performance of ST-Hedging is comparable with the leading-order approximation, and can also accurately learn the trading strategy when it is close to maturity. Moreover, the training time for the ST-Hedging algorithm is significantly smaller than the Deep Hedging algorithm, with better scalability of the number of stocks in the market. 
The training and fine tuning time for multiple assets is still 1-2 hours, which shows the scalability of ST-Hedging algorithm with  dimensions.
\begin{figure}[H]
    \centering
    \includegraphics[width=0.45\linewidth]{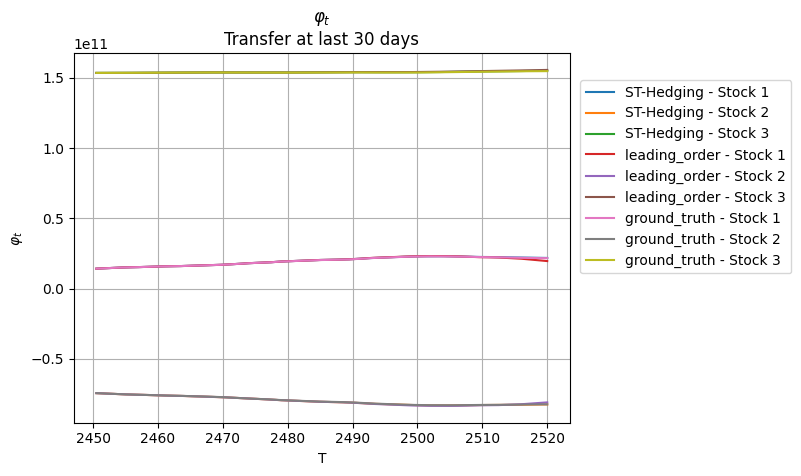}
    \includegraphics[width=0.45\linewidth]{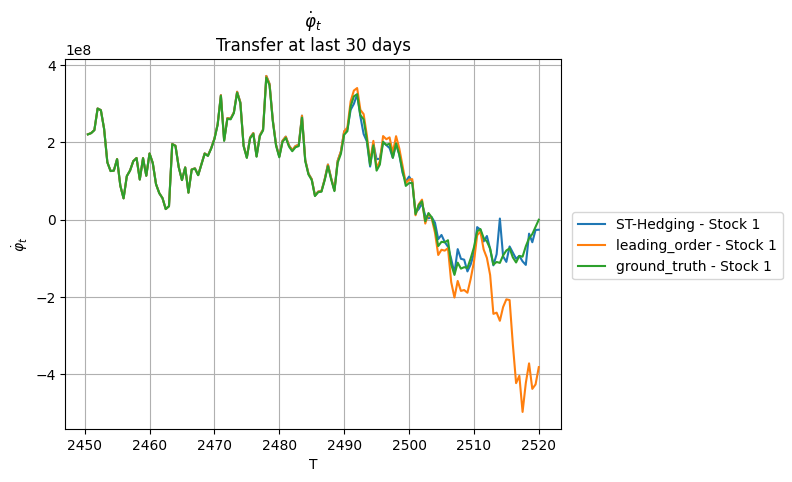}
    \includegraphics[width=0.45\linewidth]{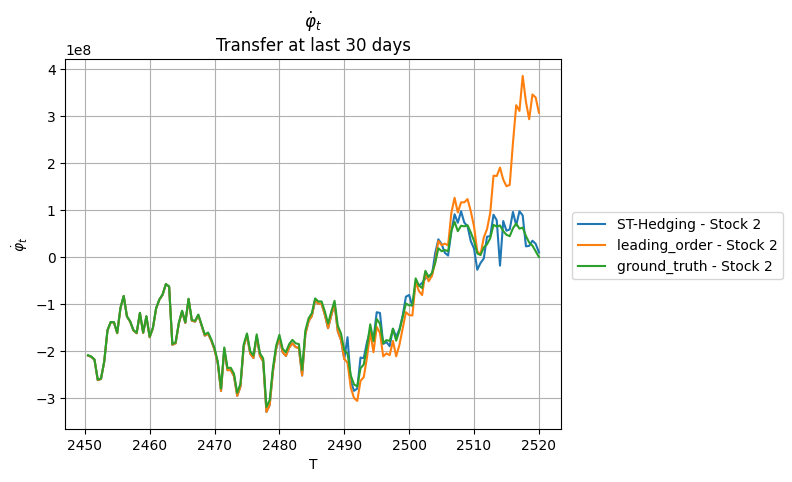}
    \includegraphics[width=0.45\linewidth]{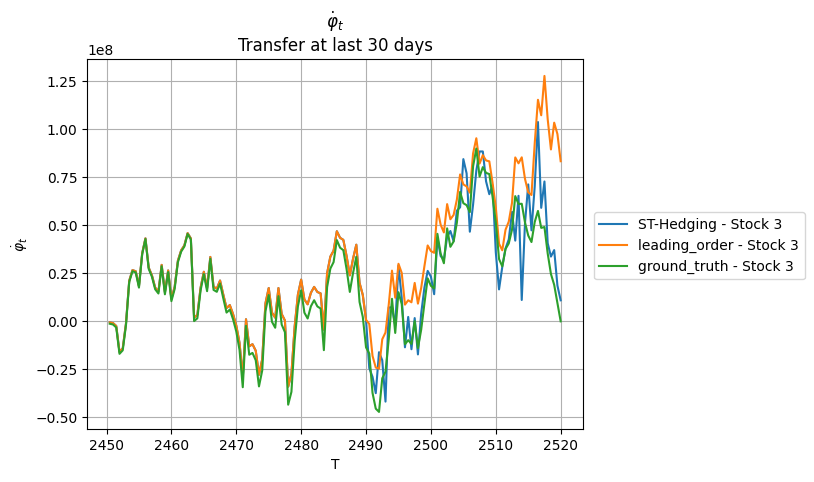}
    \caption{{Optimal trading positions} $\varphi$ (upper left panel) and {optimal trading rates} $\dot{\varphi}$ for trading horizon $T = 2520$ days with calibrated parameters under $q=2$ costs for 3 stocks using ST-Hedging.}
    \label{fig: Multi_q=2_TR=2520}
\end{figure}
\begin{table}[H]
\centering
\begin{tabular}{||c c c||} 
 \hline
 Method &  $J_T (\dot\varphi) \pm \mathrm{std}$ &$\E[|\dot\varphi_T|^2/s^2]$\\ [0.5ex] 
 \hline\hline  
ST-Hedging &	$1.60 \times 10^{11} \pm 3.58 \times 10^{6}$ & $(2.4, 4.7, 3.6)\times 10^{-9}$   \\
 Leading-order Approximation & $1.60 \times 10^{11} \pm 3.58 \times 10^{6}$ & $(4.5, 3.5, 1.6)\times 10^{-7}$  \\
 Ground Truth & $1.60 \times 10^{11} \pm 3.58 \times 10^{6}$ & $(0.0, 0.0, 0.0)$  \\ [0.5ex] 
 \hline
\end{tabular}
\caption{Expectation and standard deviation of preference $J_T$, and mean squared error of $\dot \varphi_T/s$ for trading horizon $T = 2520$ days with calibrated parameters under quadratic costs.}
\label{table: Multi_q=2_TR=2520_pasting}
\end{table}

\subsection{Comparison of Algorithms}\label{sub: comparison}
Beyond the experiments in Section~\ref{experiment: quadratic} - \ref{exp:multi},
we implement the FBSDE solver, the Deep Hedging algorithm and the ST-Hedging algorithm and test them in various market settings. The  details and codes are available here: \href{https://github.com/xf-shi/ML-for-Transaction-Costs}{https://github.com/xf-shi/ML-for-Transaction-Costs}. 

Empirically, the appearance of the extreme value in the dynamic of the FBSDE solver may affect the calculation precision and thus jeopardize the performance. In particular, the performance and the harness of tuning of the FBSDE solvers becomes significantly worse if we increase the number of forward and backward variables with respect to multiple stocks settings, because of the coupling in the FBSDE system. 
In comparison, the Deep Hedging still works well,  which shows its good scalability with respect to the increase of  dimensions. 
A major drawback of both algorithm is they do not generalize well to long trading horizons. The FBSDE solver fails to converge even in intermediate trading time horizon, whereas the Deep Hedging algorithm still works but requires significant training time. 
To overcome the above mentioned difficulties in scalabilities, our proposed ST-Hedging algorithm fully utilizes the convenience of the leading-order formula and the accuracy of the learning-based algorithms. 
In our comparison experiments, ST-Hedging beats the performance of the leading-order formula and shows great scalability with respect to the increase of the stocks and the increase of trading time horizon, with less training time required. 
 In fact, ST-Hedging can be generalized as long as the leading-order approximation is available. These comparison results help explain the leading-order asymptotic formula and provide us with ideas on the usage of each method in practice. Moreover, it opens door to the development of model-based algorithms. 
We summarize the observations (High/Medium/Low) in Table
~\ref{table: comparison}: 
\begin{table}[H]

\centering
\begin{center}
\begin{tabular}{||l  c c c ||}
 \hline
  & FBSDE Solver & Deep Hedging & ST-Hedging \\ 
  \hline\hline
 Scalability wrt time  & Low & Medium & High \\
 Scalability wrt dimension  & Medium & High & High \\
Convergence speed  & High (if converges) & Medium & High\\
Hardness of tuning  & Low (if converges) & Medium & Low \\
Sensitivity to calculation precision & High & Medium & Medium\\
Adaptivity to multiple stocks  & Medium & High & High\\
Adaptivity to  general utility functions  & Low & High & High\\
 Adaptivity to general market dynamics & Medium & High & Medium \\
\hline
\end{tabular}
\caption{Comparison of Learning-Based Algorithms} \label{table: comparison}
\end{center}

\end{table}


\section{Conclusion and Looking-forward}\label{conclusion}

This paper studies the optimal  hedging problems in frictional markets with general convex transaction costs on the trading rates, and propose the ST-Hedging algorithm to fully utilize the leading-order asymptotic formulas and the deep learning-based algorithms. 
Under the smallness assumption on the magnitude of the transaction costs, the leading-order approximation of the optimal trading speed can be identified through the solution to a nonlinear SDE. Models with arbitrary state dynamics generally lead to a nonlinear FBSDE system, but we can still solve the optimization numerically through learning-based algorithms. 
We implement the FBSDE solver, the  Deep Hedging algorithms and the ST-Hedging algorithms and compare their performances with the leading-order approximation. The leading-order approximation works well under long trading horizons, while it deviates from the ground truth by a significant amount only shortly before maturity. The FBSDE solver only performs well under short trading horizons, while it fails to work in long trading horizons. In contrast, the Deep Hedging algorithm has stable and reliable performance for short and intermediate trading horizons, while it starts to get unstable when the trading horizon is getting large. In addition, the hyperparameter tuning of the Deep Hedging algorithm becomes inefficient when the model has more sophisticated dynamics. With the development of transfer learning concept and the initial choice given by $O(\sqrt{\lambda}/T)$, the proposed ST-Hedging algorithm fully utilizes the leading-order asymptotic formula and learning-based algorithms. 
Indeed, ST-Hedging shows the best performance in all our experiments, with a drastically reduced training time. 
Moreover, the success of ST-Hedging algorithm shows great potential in  model-based algorithms, to leverage the knowledge from the domain experts and the accuracy of learning-based algorithms. 
This work is also a preliminary documents for learning-based numerical solution to frictional equilibrium models. 

\appendix
\bibliographystyle{abbrv}
\bibliography{reference_ML.bib}

\section{General Asymptotics Results}\label{sec:asymptotic}
As already emphasized above, a general existence proof for the FBSDE system~\eqref{eq:dphidyn} - \eqref{eq:BSDEY} remains a challenging open problem. 
To obtain tractable results with the general transaction costs under Assumption~\ref{cond:cost} and obtain explicit approximation trading strategies as in~\cite{caye.al.18,
guasoni.weber.18}, we focus on the financial market with the following assumptions:
\begin{assumption}\label{assump:market}
\begin{enumerate}[label=(\roman*)]
\item The processes $\Lambda$ and $\sigma^2$ are bounded away from zero; 
\item The processes $\bar{a}$, $\bar{b}$, $\Lambda$, and $\sigma$ are essentially bounded It\^o processes.
\end{enumerate}
\end{assumption}

In this section, we establish the asymptotic optimal strategy 
 for the frictional mean-variance preference~\eqref{general:regular goal}. 
For better readability, we introduce the construction of optimal strategy under Assumption~\ref{assump:market} in Appedix~\ref{formal approximations}. 
The proof for the main approximation result is in Appendix~\ref{proof}. 

\subsection{Asymptotically Optimal Trading Strategies}\label{formal approximations}
We show that,  under Assumption~\eqref{assump:market}, the smallness assumption on the transaction costs level $\lambda$ should also be  a relative quantity with respect to the trading time horizon $T$. 
The following two results are the main ingredients for the asymptotic trading strategy.

\paragraph{A nonlinear ODE}
The first ingredient to cook up the leading-order approximation is the solution to a nonlinear ODE, which is also essential to the analysis of~\cite{gonon2021asset}, Lemma 3.4 and the formal analysis of~\cite{shi2020equilibrium}, Lemma 2.5. 

\begin{lemma}\label{ODE}
Suppose the instantaneous transaction cost $G$ satisfies Assumption~\ref{cond:cost}. Then the ordinary differential equation
\begin{equation*}
 (G')^{-1}\left(\frac{g(x;\gamma,\sigma,\bar{a},\lambda)}{\lambda}\right)g'(x;\gamma,\sigma,\bar{a},\lambda) + \frac{\bar{a}^2}{2} g''(x;\gamma,\sigma,\bar{a},\lambda)  = \gamma\sigma^2 x. \tag{\ref{eqn:ergodic ODE}}
\end{equation*}
has a unique solution $g$ on $\R$ such that $xg(x)\leq 0$ for all $x\in\R$. Moreover, $g$ is odd, non-increasing on $\R$ and $g$ satisfies the growth conditions
\begin{equation}\label{conditions}
\lim_{x\to -\infty} \frac{g(x;\gamma,\sigma,\bar{a},\lambda)}{\lambda(G^*)^{-1}(\frac{\gamma\sigma^2}{2\lambda}x^2)} = 1, \qquad
\lim_{x\to +\infty} \frac{g(x;\gamma,\sigma,\bar{a},\lambda)}{\lambda(G^*)^{-1}(\frac{\gamma\sigma^2}{2\lambda}x^2)} = -1,
\end{equation}
where $G^*$ is the Legendre transform of $G$. 
\end{lemma}

\begin{remark}
If the time horizon $T$ is large, then the solution to the optimal trading strategy should become stationary. Such a stationary solution should in turn solve the  essential nonlinear ODE~\eqref{eqn:ergodic ODE}. 
 Far from the terminal time $T$, it is natural to expect that the correct solution is still identified by the same growth condition in the space variable as~\eqref{conditions}.

For power functions $G(x) = |x|^q /q$, $q\in(1,2]$, the Legendre transform is 
\begin{align*}
G^*(x) &= \sup_{y} \{xy - G(y)\}
= x (G')^{-1} (x) - G\left( (G')^{-1} (x)\right)
= |x|^p/p,
\end{align*}
where $p = q/(q-1)$ is the conjugate of $q$. In this case, with proper inner and outer rescaling coefficients, ~\eqref{eqn:ergodic ODE} is exactly the \emph{same} ODE which plays an important role in Lemma 19 and Lemma 21 in~\cite{guasoni.weber.18} and Lemma 3.1 in~\cite{caye2017trading}.

\end{remark}

\paragraph{A fast mean-reverting SDE} The second ingredient is the existence and uniqueness of a strong solution to a fast mean-reverting SDE. 
\begin{lemma}\label{lem:sde}
Let $g$ be the solution to the ODE~\eqref{eqn:ergodic ODE} from Lemma~\ref{ODE}. There exists a unique strong solution of the SDE
\begin{align}\label{eqn:ergodic SDE}
d\Delta_t &= \left((G')^{-1}\left(\frac{g(\Delta_t;\gamma,\sigma_t,\bar{a}_t,\lambda\Lambda_t)}{\lambda\Lambda_t}\right) - \bar{b}_t\right)dt -\bar{a}_t dW_t, \quad 
\Delta_0 = \varphi_{0-} +\frac{\xi_0}{\sigma_0}-\frac{\mu_0}{\gamma\sigma^ 2_0}.
\end{align}
Moreover, this process is a recurrent diffusion.
 \end{lemma}

\begin{remark}
When $G(x) = x^2/2$ and $\bar{b}$, $\bar{a}$, $\sigma$ and $\Lambda=1$ are all constants, the solution to~\eqref{eqn:ergodic ODE} is $g(x;\gamma,\sigma,\bar{a},\lambda) = -\sqrt{\gamma\sigma^2 \lambda}x$, and the dynamics ~\eqref{eqn:ergodic SDE} becomes 
\begin{align}\label{eq:OU}
d\Delta_t = -\left(\sqrt{\frac{\gamma \sigma^2}{\lambda}}\Delta_t-\bar{b}\right)dt -\bar{a} dW_t, \qquad 
\Delta_0 = \varphi_{0-} +\frac{\xi_0}{\sigma}-\frac{\mu_0}{\gamma\sigma^ 2}. 
\end{align}
This is an Ornstein-Uhlenbeck process, which is mean-reverting. In general, the requirement of  $xg(x)\leq 0$ ensures that the dynamic~\eqref{eqn:ergodic SDE} is indeed mean-reverting and converges to an ergodic limit. 
\end{remark}

With these two ingredients on hand, we now present our first results in the following theorem:
\begin{theorem}\label{thm:main}
Let $g$ be the solution to~\eqref{eqn:ergodic ODE} and $\left(\Delta_t\right)_{t\geq 0}$ the solution  to~\eqref{eqn:ergodic SDE}. 
Then  under Assumption~\ref{cond:cost} and Assumption~\ref{assump:market}, 
  for all competing admissible strategies $\dot\psi$, we have 
\begin{align*}
 J_T(\dot\psi) \leq  J_T\left( (G')^{-1}\left(\frac{g(\Delta;\gamma,\sigma,\bar{a},\lambda\Lambda)}{\lambda\Lambda}\right) \right) + O\left(\frac{\sqrt{\lambda}}{T}\right) +O\left(\sqrt{\lambda}\right). 
\end{align*}
\end{theorem}

Theorem~\ref{thm:main} shows that,  under Assumption~\ref{assump:market} the smallness is not only an absolute quantity on $\sqrt{\lambda}$, but also  on the relative quantity $\sqrt{\lambda}/T$, i.e. the smallness of $\lambda$ should also be relative quantity with respect to the trading time horizon. 
Notice that the order of the smallness is derived as a coarse upper bound for every transaction costs function $G$ that satisfies Assumption~\ref{cond:cost}. In fact, with the specific form of the transaction costs $G$, we can have finer estimations, as in the example of quadratic costs shown in Corollary~\ref{quadratic case}. For the general power costs case and proportional costs case, we refer the reader to the discussion of Theorem 3.3 in~\cite{caye2017trading} and Theorem 4.2 in~\cite{gonon2021asset}.  

\subsection{Proof of Section~\ref{formal approximations}}\label{proof}

The proof of Lemma~\ref{ODE} follows the same procedure as the proof of Lemma 3.4 in~\cite{gonon2021asset}, and Lemma~\ref{lem:sde} follows the same procedure as the proof of Lemma 3.5 in~\cite{gonon2021asset}. 

Here we provide some auxiliary results on the function $g$ from~\eqref{eqn:ergodic ODE}.
\begin{corollary}\label{ode results}
Let $g$ be the solution to~\eqref{eqn:ergodic ODE} from Lemma~\ref{ODE}.
Then the following holds:
\begin{enumerate}
\item There exists a constant $C_G>0$ that only depends on $G$ such that for all $x\in\R$, 
\begin{align}\label{est:g}
|g(x;\gamma,\sigma,\bar{a},\lambda)| \leq C_G \sqrt{\lambda}\left(\sqrt{\lambda} + \sqrt{\gamma\sigma^2} |x|\right). 
\end{align}
\item There exists a constant $K_G>0$ that only depends on $G$ such that for all $x\in\R$,
\begin{align}\label{est:g'}
|g'(x;\gamma,\sigma,\bar{a},\lambda)| \leq \sqrt{\gamma\sigma^2\lambda} K_G. 
\end{align}
\end{enumerate}
\end{corollary}

\begin{corollary}\label{sde results}
Let $g$ be the solution to~\eqref{eqn:ergodic ODE} from Lemma~\ref{ODE} and let the process $\Delta$ be the strong solution to~\eqref{eqn:ergodic SDE} from Lemma~\ref{lem:sde}. 
\begin{enumerate}
\item We have the following uniform moment bounds
\begin{align}\label{uniform bound on moments}
\sup_{T\geq 0}\E\left[ |\Delta_T|^k\right] <\infty, \qquad \mbox{for all} \quad k\in\mathbb{N}.
\end{align}
\item There exists $M>0$, such that for an arbitrary process $X$ with dynamic 
$$
dX_t = \mu_t^X dt + \sigma_t^X dW_t,
$$
the following inequality holds a.s.:
\end{enumerate}
\begin{align}\label{ass saver}
&\Big|
 \int_0^T \left(\gamma\sigma_t^2\Delta_tX_t+\mu_t^Xg(\Delta_t;\gamma,\sigma_t,\bar{a}_t,\lambda\Lambda_t) + \sigma^X_t\bar{a}_t g'(\Delta_t;\gamma,\sigma_t,\bar{a}_t,\lambda\Lambda_t) \right)dt 
 \notag\\
&\qquad + \int_0^TX_t\bar{a}_t g'(\Delta_t;\gamma,\sigma_t,\bar{a}_t,\lambda\Lambda_t)dW_t -g(\Delta_T;\gamma,\sigma_T,\bar{a}_T,\lambda\Lambda_T) X_T
\Big| 
\leq \sqrt{\lambda}M\int_0^T |X_t| dt. 
\end{align}
\end{corollary}
\vspace{10pt}
\begin{proof}[Proof of Theorem~\ref{thm:main}]
With the strategy,  we write
$$
\hat\varphi_t 
= \varphi_{0-} + \int_0^t \left(G'\right)^{-1}\left(\frac{g(\Delta_u;\gamma,\sigma_u,\bar{a}_u,\lambda\Lambda_u)}{\lambda\Lambda_u}\right)\ du 
= \varphi_{0-} + \bar\varphi_t  +\Delta_t - \left(\frac{\mu_0}{\gamma\sigma^ 2} - \frac{\xi_0}{\sigma} +\Delta_0\right) =  \bar\varphi_t + \Delta_t ,
$$
hence with~\eqref{eq:strat}, we have 
\begin{align}\label{eq:1}
\gamma\sigma^2\Delta_t = \gamma\sigma^2\left(\hat\varphi_t - \bar\varphi_t\right) = \gamma\sigma\left(\sigma\hat\varphi_t +\xi_t\right)- {\mu_t}.
\end{align}
Consider a competing admissible strategy $\psi$ and, to ease notation, set 
\begin{align*}
\dot{\theta}_t = \dot{\psi}_t -\left(G'\right)^{-1}\left(\frac{g(\Delta_t;\gamma,\sigma_t,\bar{a}_t,\lambda\Lambda_t)}{\lambda\Lambda_t}\right), 
\end{align*}
hence
\begin{align*}
\theta_t = \int_0^t \dot{\psi}_u - \left(G'\right)^{-1}\left(\frac{g(\Delta_u;\gamma,\sigma_u,\bar{a}_u,\lambda\Lambda_u)}{\lambda\Lambda_u}\right) du =  \psi_t - \hat\varphi_t. 
\end{align*}
Equation~\eqref{eq:1} and the convexity of $G$ yield
\begin{align}
&\qquad J_T(\dot{\psi}) - J_T \left( (G')^{-1}\left(\frac{g(\Delta;\gamma,\sigma,\bar{a},\lambda\Lambda)}{\lambda\Lambda}\right) \right)
\notag\\&= \frac{1}{T} \E\left[ \int_0^T\theta_t \mu_t - 
\frac{\gamma}{2}\theta_t (\psi_t \sigma_t+ \hat\varphi_t\sigma_t +2\xi_t)\sigma_t
+\lambda\Lambda_t\left(G\left(\left(G'\right)^{-1}\left(\frac{g(\Delta_t;\gamma,\sigma_t,\bar{a}_t,\lambda\Lambda_t)}{\lambda\Lambda_t}\right)\right)-G(\dot{\psi}_t) \right)\; dt \right] \notag\\
&\leq \frac{1}{T}\E\left[ \int_0^T-\frac1 2 \gamma \left(\theta_t  \sigma_t\right)^2 + \theta_t \big(\mu_t - \gamma(\hat\varphi_t\sigma_t + \xi_t)\sigma_t\big)+\lambda\Lambda_t G'\left(\left(G'\right)^{-1}\left(\frac{g(\Delta_t;\gamma,\sigma_t,\bar{a}_t,\lambda\Lambda_t)}{\lambda\Lambda_t}\right)\right)  \dot\theta_t \; dt\right] \notag\\
&= \frac{1}{T}\E\left[ \int_0^T-\frac1 2 \gamma \left(\theta_t \sigma_t \right)^2 - \gamma \theta_t \sigma^2_t\Delta_t  - g(\Delta_t;\gamma,\sigma_t,\bar{a}_t,\lambda\Lambda_t)\dot{\theta}_t  \; dt\right]. \label{eq:comp}
\end{align}
We now analyze the terms on the right-hand side of~\eqref{eq:comp}. The inequality~\eqref{ass saver} from Lemma~\ref{sde results}
in turn yields
\begin{align}\label{eq:IBP}
&\E\left[\int_0^T \left(\gamma \theta_t \sigma^2_t\Delta_t  + g(\Delta_t;\gamma,\sigma_t,\bar{a}_t,\lambda\Lambda_t)\dot{\theta}_t\right)dt\right]
\notag\\&\qquad\qquad\qquad\qquad
\geq \E[g(\Delta_T;\gamma,\sigma_T,\bar{a}_T,\lambda\Lambda_T)\theta_T] - \sqrt{\lambda} M \E\left[\int_0^T |\theta_t|dt\right]
\end{align}
Here, the local martingale part is a true martingale. Indeed, by H\"older's inequality, the integrability condition  $\varphi\sigma\in\mathbb{H}^2$ and the boundedness of $g'$ established in Corollary~\ref{ode results},
\begin{align*}
\E\left[ \int_0^t |\theta_u g'(\Delta_u;\gamma,\sigma_u,\bar{a}_u,\lambda\Lambda_u)|^2 du \right] 
&\leq \gamma\lambda K_G^2\E\left[ \int_0^t\sigma^2_u \theta_u^{2}du \right]<\infty.
\end{align*}
Also taking into account that 
$$
\left| g(\Delta_t;\gamma,\sigma_t,\bar{a}_t,\lambda\Lambda_t) \right| \leq \sqrt{\gamma\sigma^2_t\lambda\Lambda_t} C_G |\Delta_t| + \lambda\Lambda_t C_G,
$$
we can therefore use~\eqref{eq:IBP} to replace the second and the third terms on the right-hand side of~\eqref{eq:comp}, obtaining
\begin{align*}
&J_T(\dot{\psi}) - J_T \left( (G')^{-1}\left(\frac{g(\Delta;\gamma,\sigma,\bar{a},\lambda\Lambda)}{\lambda\Lambda}\right) \right)
\\&\qquad\qquad\qquad\qquad\leq - \frac{1}{T}\E[g(\Delta_T;\gamma,\sigma_T,\bar{a}_T,\lambda\Lambda_T)\theta_T] - \E\left[ \int_0^T \frac{\gamma}{2} \left(\theta_t\sigma_t\right)^2 dt \right] + \frac{\sqrt{\lambda}M}{T}\int_0^T \E[|\theta_t|]dt.
\end{align*}
The Cauchy-Schwartz inequality yields
\begin{align*}
\big|\E[g(\Delta_T;\gamma,\sigma_T,\bar{a}_T,\lambda\Lambda_T)\theta_T]\big|
&\leq \left(\E[g(\Delta_T;\gamma,\sigma_T,\bar{a}_T,\lambda\Lambda_T)^2]\E[\theta_T^2]\right)^{1/2} 
\\& \leq \left(\E[2g(\Delta_T;\gamma,\sigma_T,\bar{a}_T,\lambda\Lambda_T)^2](\E[(\hat\varphi_T)^2] + \E[(\psi_T)^2])\right)^{1/2}
\\&\leq 2C_G \sqrt{\lambda} \left(\E[(\hat\varphi_T)^2] + \E[(\psi_T)^2]\right)^{1/2} \left( \gamma\sigma^2  \E[|\Delta_t|^2]+  \lambda\right)^{1/2}. 
\end{align*}
Moreover, it follows that 
$$
\frac{1}{T}\E[|\hat\varphi_T|^2] =\frac{2}{T}\left(\E[|\bar\varphi_T|^2] + \E[|\Delta_T|^2]\right) \leq 2\sup_{T>0} \frac{1}{T}\left(\E[|\bar\varphi_T|^2] + \E[|\Delta_T|^2]\right) <\infty. 
$$
 Together with the transversality condition~\eqref{assump:vanishing}, it follows that
\begin{align*}
\frac{1}{T} \left|\E[g(\Delta_T;\gamma,\sigma_T,\bar{a}_T,\lambda\Lambda_T)\theta_T] \right|
\leq  \frac{2C_G \sqrt{\lambda}}{T}\left(\E[(\hat\varphi_T)^2] + \E[(\psi_T)^2]\right)^{1/2} \left( \gamma\sigma^2  \E[|\Delta_t|^2]+  \lambda \right)
= O\left(\frac{\sqrt{\lambda}}{T}\right),
\end{align*}
and again by H\"older's inequality, 
$$
\frac{1}{T}\E\left[\int_0^T |\theta_t| dt \right] \leq \left(\frac{1}{T} \E\left[\int_0^T\theta_t^2 dt\right]\right)^{1/2} 
\leq 2\left(\sup_{T>0} \frac{1}{T}\E\left[\int_0^T \hat\varphi_t^2 + \psi_t^2 dt\right]  \right)^{1/2}. 
$$
Therefore, the trading rate $\dot{\varphi}$ is indeed asymptotically optimal as:
\begin{align*}
 &J_T(\dot{\psi}) - J_T \left( (G')^{-1}\left(\frac{g(\Delta;\gamma,\sigma,\bar{a},\lambda\Lambda)}{\lambda\Lambda}\right) \right)
\\&\qquad\qquad\qquad\leq  - \E\left[ \int_0^T \frac{\gamma}{2} \left(\theta_t\sigma_t\right)^2 dt \right]- \frac{1}{T}\E[g(\Delta_T;\gamma,\sigma_T,\bar{a}_T,\lambda\Lambda_T)\theta_T]  + \frac{\sqrt{\lambda}M}{T}\int_0^T \E[|\theta_t|]dt\\
&\qquad\qquad\qquad=- \E\left[ \int_0^T \frac{\gamma}{2} \left(\theta_t\sigma_t\right)^2 dt \right]+O\left(\frac{\sqrt{\lambda}}{T}\right) +O\left(\sqrt{\lambda}\right)\\
&\qquad\qquad\qquad\leq
O\left(\frac{\sqrt{\lambda}}{T}\right) +O\left(\sqrt{\lambda}\right). 
\end{align*}
\end{proof}

\section{Asymptotic Results for Quadratic Costs}\label{sec:asymptotic quadratic}
When the transaction costs is considered to be quadratic, i.e. $G(x) = x^2/2$, we have a linear function as $\left(G'\right)^{-1} (x) = x$. Hence the forward equation~\eqref{eq:dphidyn} becomes \emph{linear} with respect to the backward component $Y$, and the existence and uniqueness can be established as in~\cite{delarue2002existence,kohlmann.tang.02}, provided the coefficients satisfies certain regularities. However, for general function $G$ satisfying Assumption~\eqref{cond:cost},  the generator for the forward component is not globally Lipschitz hence no general theory is available for the FBSDE system~\eqref{eq:dphidyn} - \eqref{eq:BSDEY}. 

\subsection{A Concrete Example}\label{ss: exact}

As already emphasized above, a general existence proof for the FBSDE system~\eqref{eq:dphidyn} - \eqref{eq:BSDEY} remains a challenging open problem.   Let us just briefly sketch how the nonlinear FBSDE~\eqref{eq:dphidyn} - \eqref{eq:BSDEY} reduces to a nonlinear PDE, 
with the following assumptions on the market:
\begin{assumption}\label{assump:quadratic}
\begin{enumerate}[label=(\roman*)]
\item the frictionless strategy satisfies that $\bar b_t=0$ and $\bar a_t = \bar a$;
\item the volatility process of the stock price remain constant $\sigma>0$ in the models with and without transaction costs;
\item the cost parameter is constant ($\Lambda=1$).
\end{enumerate}
\end{assumption}

Under Assumption~\ref{assump:quadratic}, the forward-backward system~\eqref{eq:dphidyn} - \eqref{eq:BSDEY} in turn becomes autonomous, 
\begin{align}
&d\Delta\varphi_t = (G')^{-1}\left(\frac{Y_t}{\lambda}\right) dt - \bar{a}dW_t, &\Delta\varphi_0 &= \varphi_{0-} +\frac{\xi_0}{\sigma}-\frac{\mu_0}{\gamma\sigma^ 2}, \label{eq:fwd}\\
&dY_t = \gamma\sigma^2\Delta\varphi_t dt + Z_t dW_t,  &Y_T &= 0. \label{eq:bwd2}
\end{align}

When the transaction costs is quadratic, i.e. $ G(x) = x^2/2$, we can create the optimal strategy explicitly. Accordingly, the approximation can be made more precise, and we summarize the results as follows:  
\begin{corollary}\label{quadratic case}
For $G(x) = x^2/2$, define the  strategy
\begin{align}\label{optimal: quadratic}
\dot\varphi_t = - \sqrt{\frac{\gamma\sigma^2}{\lambda}}\tanh\Delta\varphi_t. 
\end{align}
Then under Assumption~\ref{assump:quadratic},   for all competing admissible strategies $\dot\psi$, we have 
\begin{align*}
\sup_{\dot\psi} J_T(\dot\psi) = J_T(\dot\varphi) 
=  J_T\left( -\sqrt{\frac{\gamma\sigma^2}{\lambda}} \Delta \right) + O\left(\frac{\lambda}{T}\right)
=  J_T\left( -\sqrt{\frac{\gamma\sigma^2}{\lambda}} \Delta \right) + o\left(\frac{\sqrt{\lambda}}{T}\right),
\end{align*}
where $\Delta$ is the following Ornstein–Uhlenbeck process:
$$
d\Delta_t = -\sqrt{\frac{\gamma \sigma^2}{\lambda}}\Delta_tdt -\bar{a} dW_t, \qquad 
\Delta_0 = \varphi_{0-} +\frac{\xi_0}{\sigma}-\frac{\mu_0}{\gamma\sigma^ 2}. 
$$
\end{corollary}

\begin{remark}
The optimality of  $\dot\varphi$ by~\eqref{optimal: quadratic} is derived in Theorem 4.5~\cite{muhlekarbe2021equilibrium}. The approximation order is a straight comparison between the candidate trading rate given by $-\sqrt{{\gamma \sigma^2}/{\lambda}}\ \Delta$ and the optimal trading rate $\dot\varphi$. 

When there is no liquidity risk, i.e. when $\Lambda=1$ throughout the trading horizon, the smallness assumption on the liquidity $\lambda$ is purely a relative quantity, comparing to the trading horizon $T$. Here with quadratic trading costs, the approximation is $\lambda/T$, which is finer than the overall approximation $\sqrt{\lambda}/T$. 
\end{remark}

\end{document}